\documentclass{IEEEtran}
\usepackage{float}
\usepackage{graphicx}
\usepackage{clrscode}
\usepackage{stfloats}
\usepackage{paralist}
\usepackage{tikz}
\usepackage{soul} 
\usepackage{color}
\usepackage{cite}
\usepackage{multicol}
\usepackage{comment}
\usepackage{amsmath, amssymb, mathrsfs, amsfonts}
\usepackage{amsthm}
\usepackage{mathtools}
\usepackage{epsfig, epstopdf}
\usepackage{algorithm}
\usepackage{enumitem}
\usepackage{algorithmic}
\usepackage{hyperref}
\usepackage{lipsum}
\usepackage{pbox}
\usepackage{array}
\usepackage{tablefootnote}
\usepackage{subfigure}

\usetikzlibrary{arrows,automata}

\allowdisplaybreaks

\newtheorem{lemma}{Lemma}

\newcommand{\RNum}[1]{\uppercase\expandafter{\romannumeral #1\relax}}

\begin{document}

\title{On Stochastic Fundamental Limits in a Downlink Integrated
Sensing and Communication Network}

\author{
		\IEEEauthorblockN{Marziyeh Soltani, Mahtab Mirmohseni, \textit{Senior Member, IEEE}, and Rahim Tafazolli, \textit{Senior Member, IEEE} \\
		\vspace*{0.5em}
			}\thanks{Part of this paper has been presented at the 2023 IEEE Global Communications
Conference workshop, 4-8 December 2023,  Kuala Lumpur, \cite{soltani2023outage}.\\
The authors are with 5/6GIC, the Institute for Communication Systems,
University of Surrey, GU2 7XH Guildford, U.K. (e-mail: m.mirmohseni@
surrey.ac.uk, r.tafazolli@surrey.ac.uk, m.soltani@surrey.ac.uk)}}

\maketitle
\begin{abstract}
This paper aims to analyze the stochastic performance of a multiple input multiple output (MIMO) integrated
sensing and communication (ISAC) system in a downlink scenario, where a base station (BS) transmits a dual-functional radar-communication (DFRC) signal matrix, serving the purpose of transmitting communication data to the user while simultaneously sensing the angular location of a target. The channel between the BS and the user is modeled as a random channel with Rayleigh fading distribution, and the azimuth angle of the target is assumed to follow a uniform distribution. Due to the randomness inherent in the network, the challenge is to consider suitable performance metrics for this randomness. To address this issue, for users, we employ the user's rate outage probability (OP) and ergodic rate, while for target, we propose using the OP of the Cramér-Rao lower bound (CRB) for the angle of arrival and the ergodic CRB. We have obtained the expressions of these metrics for scenarios where the BS employs two different beamforming methods. Our approach to deriving these metrics involves computing the probability density function (PDF) of the signal-to-noise ratio for users and the CRB for the target. We have demonstrated that the central limit theorem provides a viable approach for deriving these PDFs. In our numerical results, we demonstrate the trade-off between sensing and communication (S \& C) by characterizing the region of S \& C metrics and by obtaining the Pareto optimal boundary points, confirmed with simulations.
\end{abstract}
\begin{IEEEkeywords}
Integrated Sensing and Communications, Performance analysis, Outage tradeoff, CRB, Randomness, Rayleigh fading
\end{IEEEkeywords}
\section{Introduction}\label{introduction}
\IEEEPARstart{R}{ecently}, there has been a significant change in the way radar sensing and wireless communication systems are approached. This shift, known as integrated sensing and communications (ISAC), is anticipated to have a crucial impact on the advancement of future wireless networks \cite{SeventyYearsofRadarandCommunications}. ISAC is required in various emerging applications, including vehicular communication, internet of things (IoT) applications like smart city infrastructure (such as traffic and speed monitoring, video surveillance), and smart industry \cite{Integratedtoward}. By leveraging shared resources such as wireless spectrum, hardware platforms, and energy consumption, ISAC offers significant advantages for both sensing and communication (S \& C). Due to this resource sharing, there is an inherent tradeoff in ISAC, and to design an efficient ISAC system, the fundamental
communication-sensing performance tradeoff should be fully
understood \cite{ASurveyonFundamentalLimits}. 

\textcolor{blue}{The fundamental limits of ISAC systems has been explored from various perspectives: either from pure information-theoretic (IT) approach or different sensing design metrics.
\\
Several works have provided performance limits for ISAC systems from an IT perspective, including \cite{ASurveyonFundamentalLimits,2018JointStateSensingandCommunicationOptimalTradeoffforaMemorylessCase,2019JointStateSensingandCommunicatiooverMemorylessMultipleAccessChannels,2020JointSensingandCommunicationoverMemorylessBroadcastChannels,2022AnInformation-TheoreticApproachtoJointSensingandCommunication}. The authors in \cite{2018JointStateSensingandCommunicationOptimalTradeoffforaMemorylessCase,2022AnInformation-TheoreticApproachtoJointSensingandCommunication} characterize the capacity-distortion region, which is extended to multiple access channels in \cite{2019JointStateSensingandCommunicatiooverMemorylessMultipleAccessChannels} and to broadcast channels in \cite{2020JointSensingandCommunicationoverMemorylessBroadcastChannels}. These studies utilize the state-dependent channel approach to establish an abstract model for ISAC.
\\
Beyond the IT perspective, several works have analyzed the trade-off region between S \& C using different sensing performance metrics and varying levels of randomness in the system model \cite{aframeworkformutualinformation,TowardDualfunctionalRadarCommunicationSystems,MUMIMOCommunicationsWithMIMORadar,JointTransmitBeamformingforMultiuser,optimaltransmitbeamformingintegrated,CrameRaoBoundOptimizationforJoint,fromtorchtoprojector,MIMOIntegratedSensingandCommunicationCRBRateTradeoff,fundamentalcrbratetradeoffmultiantenna,OnthePerformanceofUplinkandDownlink,PerformanceAnalysisandPowerAllocationforCooperative,NOMAISACPerformanceAnalysisandRateRegion,Aunifiedperformanceframeworkfor,MIMOISACPerformanceAnalysis,PerformanceAnalysioftheFullDuplexJoint,onthefundementaltradeoff,networklevelintegratedsensingcommunication,coverageandrateofjointcommunication}. \cite{aframeworkformutualinformation} considers estimation mutual information as a sensing metric. However, transmit beampattern and Cramér-Rao Bound (CRB) are considered as the most common metrics for sensing performance evaluation.
\\
Several studies, such as \cite{TowardDualfunctionalRadarCommunicationSystems,MUMIMOCommunicationsWithMIMORadar,JointTransmitBeamformingforMultiuser,optimaltransmitbeamformingintegrated}, have adopted the transmit beampattern as a sensing performance metric, where transmit signal beams are directed toward target directions to enhance the estimation or the detection. \cite{TowardDualfunctionalRadarCommunicationSystems,MUMIMOCommunicationsWithMIMORadar} examined a multi-user multiple-input-single-output (MISO) ISAC system in which information signal beams are reused for both S \& C. In this system, the transmit information beamformers are optimized to align with a sensing-oriented beampattern while ensuring communication performance. To fully exploit the degrees of freedom (DoFs) offered by multiple input multiple output (MIMO) radar, \cite{JointTransmitBeamformingforMultiuser} and \cite{optimaltransmitbeamformingintegrated} proposed transmitting dedicated radar sensing signal beams in addition to information beams, with joint optimization of the transmit information and sensing beamformers.
\\
On the other hand, CRB is another widely used sensing performance measure for estimation tasks \cite{CrameRaoBoundOptimizationforJoint,fromtorchtoprojector,MIMOIntegratedSensingandCommunicationCRBRateTradeoff,fundamentalcrbratetradeoffmultiantenna}, which characterizes the lower bound of the error for any unbiased estimator. 
\cite{CrameRaoBoundOptimizationforJoint} examined a multi-user ISAC system over a broadcast channel with a single target, where the transmit beamforming was optimized to minimize CRB of the error the estimation, subject to individual signal-to-interference-plus-noise ratio (SINR) constraints for multiple communication users (CUs). Furthermore, \cite{fromtorchtoprojector} and \cite{MIMOIntegratedSensingandCommunicationCRBRateTradeoff} investigated the CRB-rate (C-R) trade-off in a point-to-point MIMO ISAC system involving one CU and one target. In particular, \cite{MIMOIntegratedSensingandCommunicationCRBRateTradeoff} optimized the transmit covariance to characterize the complete Pareto boundary of the C-R region for the MIMO ISAC system in two specific scenarios: point and extended target models, considering a sufficiently long radar coherent processing interval. \cite{fromtorchtoprojector} explored the C-R trade-off for a more general case with finite radar coherent processing intervals. Additionally, \cite{fundamentalcrbratetradeoffmultiantenna} characterized the Pareto-optimal boundary of the C-R region in a system where a multi-antenna base station (BS) transmits common information messages to a set of single-antenna CUs while concurrently estimating the parameters of multiple sensing targets based on the echo signals.
\\
The impact of channel randomness and its statistical properties, which is the focus of this paper, has been largely overlooked in the aforementioned works \cite{aframeworkformutualinformation,TowardDualfunctionalRadarCommunicationSystems,MUMIMOCommunicationsWithMIMORadar,JointTransmitBeamformingforMultiuser,optimaltransmitbeamformingintegrated,CrameRaoBoundOptimizationforJoint,fromtorchtoprojector,MIMOIntegratedSensingandCommunicationCRBRateTradeoff,fundamentalcrbratetradeoffmultiantenna}. Recent studies have examined the probabilistic behavior of ISAC by incorporating randomness in channels. These studies can be categorized into three groups based on sensing metrics: detection probability \cite{OnthePerformanceofUplinkandDownlink,PerformanceAnalysisandPowerAllocationforCooperative, Aunifiedperformanceframeworkfor} and \cite[Section A]{Aunifiedperformanceframeworkfor}, sensing rate or SNR \cite{NOMAISACPerformanceAnalysisandRateRegion,MIMOISACPerformanceAnalysis,PerformanceAnalysioftheFullDuplexJoint,networklevelintegratedsensingcommunication,coverageandrateofjointcommunication}, and estimation error \cite{onthefundementaltradeoff}.
\\
In the first category \cite{OnthePerformanceofUplinkandDownlink,PerformanceAnalysisandPowerAllocationforCooperative, Aunifiedperformanceframeworkfor}, the sensing metric is detection probability, and the user's channel follows Rayleigh fading. The communication metrics are outage probability or achievable rate. In the second category, the sensing metric is sensing rate or SNR \cite{NOMAISACPerformanceAnalysisandRateRegion,MIMOISACPerformanceAnalysis,PerformanceAnalysioftheFullDuplexJoint,networklevelintegratedsensingcommunication,coverageandrateofjointcommunication,PerformanceofDownlinkandUplinkIntegratedSensing}. \cite{NOMAISACPerformanceAnalysisandRateRegion} used the sensing rate as the performance metric, assuming that the user channel follows Rayleigh fading. The authors in \cite{MIMOISACPerformanceAnalysis,PerformanceofDownlinkandUplinkIntegratedSensing} analyzed S \& C rates in a MIMO ISAC scenario where a BS communicates with multiple users while sensing nearby targets, with user channels experiencing Rayleigh fading and target response matrices following a complex normal distribution. In \cite{PerformanceAnalysioftheFullDuplexJoint}, the authors examined ergodic communication and radar estimation rates in a full-duplex multi-antenna BS scenario, where uplink and downlink communication occurs alongside radar sensing. In \cite{networklevelintegratedsensingcommunication}, An ISAC network was studied using stochastic geometry to balance S \& C performance, with ergodic radar rates and communication rates as metrics. Finally, \cite{coverageandrateofjointcommunication} analyzed S \& C coverage and rates in a mmWave joint communication and sensing (JSAC) network using directional beamforming, providing upper and lower bounds based on stochastic geometry. In the third category \cite{onthefundementaltradeoff}, the sensing metric is estimation error. The authors considered a general point-to-point (P2P) system setting where a random signal is emitted from a multi-antenna BS and received by both a communication user and a sensing target. They analyzed the trade-off between communication capacity and CRB, where the sensing parameter follows a certain distribution.}
\\
\textcolor{blue}{To the best of our knowledge, this paper is the first to characterize the performance of ISAC by focusing on the estimation of the target angle as the sensing task, considering the randomness in both the communication user channel vector and the target. We specifically focus on a downlink MIMO ISAC system where a multiple-antenna BS transmits a dual-functional precoding matrix. This matrix serves the dual purpose of transmitting communication data to a user while concurrently sensing the angular location of a target via the reflected echo signal. The channel and the channel estimation error between the BS and the user are modeled as two independent, circularly symmetric complex Gaussian random vectors (Rayleigh fading channel models, which may not be applicable to higher frequency bands), and the target's angle is assumed to follow a uniform distribution.\footnote{\textcolor{blue}{The approach used for deriving the sensing and communication performance metrics is applicable to any arbitrary distribution for the target's angle.}} We define the C-R region and aim to characterize its Pareto boundary, where each point on this boundary represents a trade-off, such that neither sensing nor communication can be improved without degrading the other. We consider two beamforming methods, depending on whether different signals are designed for user and radar purposes or not: Subspace Joint Beamforming (SJB) and Linear Beamforming (LB). The primary contributions of this paper can be summarized as follows:
\\
1) We reveal the S \& C trade-off in a random ISAC system by considering both ergodic and outage probability (OP) metrics, suitable for fast and slow fading scenarios, respectively. Specifically, we derive exact expressions for four key metrics in random channels under perfect and imperfect channel state information (CSI) at the BS in both SJB and LB: the user's OP and ergodic rate. For the target, we propose two novel metrics: the OP of the target, denoted as \(P(\text{CRB}(\theta) > \epsilon)\), where \(\text{CRB}(\theta)\) represents the CRB for the angle of arrival, and the ergodic CRB, denoted as \(E\{\text{CRB}\}\), which we prove provides a tighter bound than the Bayesian CRB (BCRB).\footnote{These metrics were highlighted as a future research direction in the overview paper \cite{recentadvancesandtenopenchallenges}.} In LB, we apply Costa's dirty paper coding (DPC) \cite{costa} and consider cases without DPC.
\\
2) We prove that the SJB beamforming vector is optimal for minimizing the CRB when using the same signals for S \& C. In LB, we propose using separate beamforming vectors for the user and radar signals, which, while not optimal, fully exploit the DoF of MIMO radar \cite{JointTransmitBeamformingforMultiuser,optimaltransmitbeamformingintegrated} and can outperform SJB.
\\
3) In LB, we derive upper and lower bounds, along with an approximation for the OP of the target, using the Cauchy-Schwarz inequality and the law of large numbers.
\\
4) We analyze the system's performance at two key points: optimal sensing and optimal communication, achieved by directing all BS power to the target or user. These are termed "opportunistic sensing" and "opportunistic communication," as one function operates without affecting the other. 
\\
5) In the simulation, we illustrate the S \& C trade-off by deriving the achievable OP region for the target and user, and including the time-sharing line to show ISAC's benefits. When the target OP threshold is high, JB outperforms LB by lowering OP for both target and user. However, when the target OP threshold is low and effective communication (low user OP) is needed, LB performs better \footnote{\textcolor{blue}{We note that our model focuses on a single range-Doppler bin, which refers to a specific cell in a two-dimensional grid (or matrix) used in radar signal processing to localize a target in terms of its range (distance) and Doppler shift (relative speed). Specifically, the performance metrics presented in this paper should be applied individually to each range-Doppler bin.}}.
}
\\
\textbf{Differences between this paper and recent works:} \textcolor{blue}{ Our paper focuses on the stochastic analysis of ISAC, with an emphasis on estimation as the sensing task, which is also explored in \cite{onthefundementaltradeoff,coverageandrateofjointcommunication}. Unlike \cite{onthefundementaltradeoff}, which only considers two points on the C-R boundary and ignores user channel randomness without considering a beamforming vector, we characterize the entire boundary, assuming a beamforming vector aligned with both the target and user random channels. While \cite{onthefundementaltradeoff} uses the BCRB, we derive the OP and ergodic CRB, which is tighter. Unlike \cite{coverageandrateofjointcommunication}, which focuses on delay or Doppler estimation and derives sensing rate and coverage upper and lower bounds, we focus on angle estimation, deriving exact CRB, user rate, and OP metrics, rather than bounds. We also derive the Pareto boundary to illustrate the S \& C trade-off. In contrast to \cite{coverageandrateofjointcommunication}, which ignores the target angle, does not use beamforming, and parameterizes the antenna pattern based on main lobe gain, we consider two beamforming vectors that are functions of both the random target and user channels. This makes the derivation more challenging due to the correlation between the beamforming vectors and these channels. Furthermore, unlike \cite{networklevelintegratedsensingcommunication}, which derives ergodic radar and communication rates using zero-forcing beamforming to nullify intra-cluster S \& C interference, we focus on the estimation task and CRB for sensing based on two different beamforming strategies, ensuring that the effects of S \& C are not nullified in each other's performance. We also account for both perfect and imperfect CSI at the BS, unlike the studies in \cite{onthefundementaltradeoff,networklevelintegratedsensingcommunication,coverageandrateofjointcommunication}.}

The subsequent sections of this article are organized as follows: Section \ref{systemmodel} provides the system model and introduces the beamforming methods, SJB and LB. Sections \ref{jointbeamformingperformance} and \ref{linearbeamforing} analyze the system performance for SJB and LB, respectively. Opportunistic sensing/communication approaches are provided in Section \ref{specialpoints}. Simulations and conclusions are presented in Sections \ref{simulations} and \ref{conclusion}, respectively.

\textbf{Notation:} We use lowercase letters, boldface lowercase, and boldface uppercase to denote scalar quantities, vectors, and matrices, respectively. $P(.), f_x(.), E[.]$ represent the probability, probability density function (PDF), and expectation. $\mathbf{X}^{M \times N}$, $\mathbf{X}^T, \mathbf{X}^H$, and $\mathbf{X}^*$ denote a matrix with $M$ rows and $N$ columns, its transpose, Hermitian transpose, and conjugate, respectively. The Euclidean norm of a vector is $\parallel. \parallel$, and the norm of a complex number is $|.|$. $\mathcal{CN}(.,.)$ represents a circularly-symmetric complex Gaussian distribution, $\mathcal{N}_3(\mathbf{\mu}, \mathbf{\Sigma})$ is a trivariate normal distribution with mean $\mathbf{\mu}$ and covariance matrix $\mathbf{\Sigma}$, and $F_{w,k,\lambda,s,m}(.)$ is the cumulative distribution function (CDF) of a generalized chi-square RV with parameters $w, k, \lambda, s, m$. $\mathbb{C}$ and $\mathbb{R}$ denote sets of complex and real numbers. $\overset{d}\rightarrow$ and $\overset{p}\rightarrow$ indicate convergence in distribution and probability. \text{Tr($A$)} represents the trace of matrix $A$.
\begin{figure}
    \includegraphics[scale=.38]{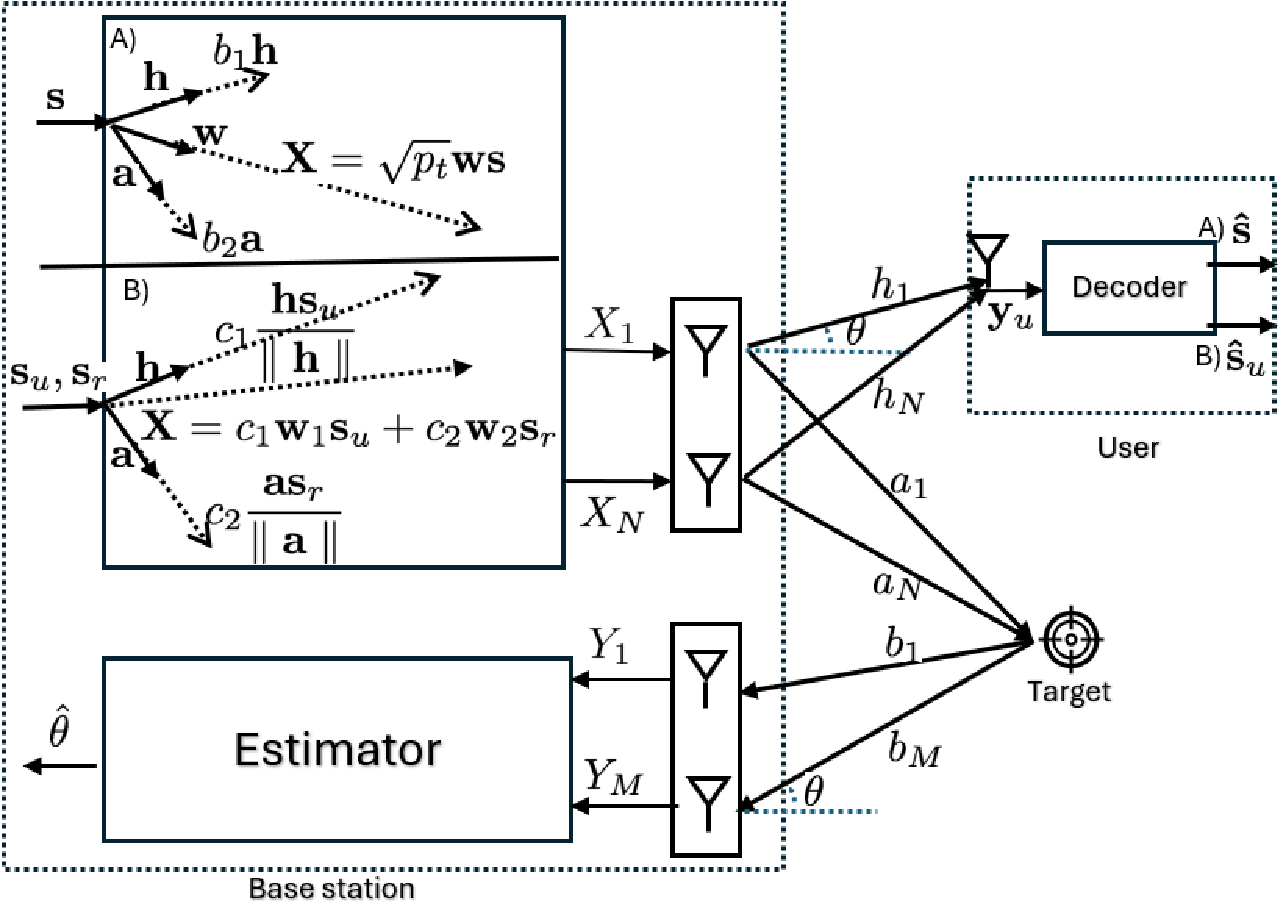}\caption{System model: (A) subspace joint beamforming (SJB), (B) linear beamforming (LB).} \label{systemmodelfig}
\end{figure}
\section{System Model}\label{systemmodel}
We consider a scenario where a BS is equipped with $N$ transmit antennas and $M$ receive antennas and aims to serve a single-antenna user in the downlink by transmitting data at a rate denoted as $R$, ensuring that the user can accurately decode the message. Simultaneously, the BS senses a target located at a distant point away from the BS, assumed to be a point target. We consider a
mono-static radar setting where both the estimator (radar receiver) and transmit antennas are located at BS and the direction of arrival
(DoA) and the direction of departure (DoD) are the same. The BS sends $\mathbf{X} \in \mathbb{C}^{N\times L}$, a narrow-band dual functional radar communication signal matrix, where $L>N$ is the length of the radar pulse/communication
frame. Channel vector between the BS and the user is $\mathbf{h}=[h_1\quad h_2 ...h_N]^T \in \mathbb{C}^{N \times 1}$, whose elements are independent and identically distributed (i.i.d.) with distribution $\mathcal{CN} (0,1)$. More precisely, element $i$-th can be expressed as $h_i=m_i+jn_i$ with $m_i$, $n_i \sim \mathcal{N}(0,\frac{1}{2})$. For all $i$, the elements $m_i$ and $n_i$ are i.i.d. \textcolor{blue}{We consider imperfect CSI at the BS. Thus, we model the actual channel, $\mathbf{h}$, for each user as \(\mathbf{\bar{h}} = \mathbf{h} + \mathbf{e}\), where \(\mathbf{\bar{h}}\) represents the estimated channel at the BS, and \(\mathbf{e}\) denotes the channel error vector. We utilize the widely-adopted Gaussian channel error model, as discussed in \cite{probabilisticallyconstrainedapproachestothedesign}. Specifically, \(e \sim \mathcal{CN}(0, \mathbf{C})\), where \(\mathbf{C}\) is a known error covariance matrix. For simplicity, we assume \(\mathbf{C} = \sigma_e^2 \mathbf{I}\) \footnote{\textcolor{blue}{The performance of the system under perfect CSI can be derived by setting $\sigma_e=0$.}}.}
Furthermore, assuming an even number of antennas, the transmit and receive array steering vectors from the BS to the target are defined as follows: 
\begin{align}
\mathbf{a}(\theta) \!\!\!&=\!\! \left[ e^{-j\pi \sin(\theta)\frac{N-1}{2}}\!\!, e^{-j\pi \sin(\theta)\frac{N-3}{2}}\!\!, \!\!\ldots, e^{j\pi \sin(\theta)\frac{N-1}{2}} \right]^T \!\!\!\!\!\!\in \mathbb{C}^{N \!\times \!1}\nonumber\\
\mathbf{b}(\theta)\!\!\!&=\!\! \left[ e^{-j\pi \sin(\theta)\frac{M-1}{2}}\!\!, e^{-j\pi \sin(\theta)\frac{M-3}{2}}\!\!, \!\! \ldots, e^{j\pi \sin(\theta)\frac{M-1}{2}} \right]^T  \!\!\!\!\!\!\in \mathbb{C}^{M \!\times\! 1}\nonumber
\end{align}
where $\theta$ represents the azimuth angle of the target relative to the BS, and it follows a uniform distribution in the interval $[0,\pi]$. We express $i$-th element of $\mathbf{a}(\theta)$ as $a_i=e^{-jf_i}$, where $f_i=\pi \sin(\theta)\frac{N-(2i-1)}{2}$. The received signal at the user is $\mathbf{y}_u=\mathbf{h}^H \mathbf{X}+\mathbf{z}_u$, where $\mathbf{z}_u\in \mathbb{C}^{1 \times L}$ is the additive white Gaussian noise
(AWGN) vector where each of its elements has the distribution of the form $\mathcal{CN} (0,\sigma^2_u)$. When the BS transmits $\mathbf{X}$ to sense the target, it receives back
the reflected echo signal matrix as $\mathbf{Y}_r=\alpha \mathbf{b}(\theta)\mathbf{a}(\theta)^H \mathbf{X}+\mathbf{Z}_r$, where \textcolor{blue}{$\alpha \in \mathcal{C}$ denotes the complex-valued channel coefficient, which is dependent on the target’s radar cross section (RCS) and the round-trip path loss.} $\mathbf{Z}_r\in \mathbb{C}^{M \times L}$ is AWGN matrix which its elements being i.i.d. and having the distribution of the form $\mathcal{CN} (0,\sigma^2_r)$.

The metrics employed for evaluating the performance of S \& C are CRB and rate, respectively. 
Given the random nature of the channels, $\text{CRB}(\theta)$ and the user's rate ($R$) are RVs, highlighting the need for defining metrics suitable for analyzing performance in random networks. Thus, we utilize two metrics for users: the user's OP, defined as $P(\text{SINR}<\gamma)$, which is suitable for slow fading scenarios, and the ergodic rate, expressed as $E[\log(1+\text{SINR})]$, suitable for fast fading. For the target, we propose two novel metrics: the ergodic CRB, defined as \( E[\text{CRB}(\theta)] \), and the outage probability (OP) of the target, represented as \( P(\text{CRB}(\theta) > \epsilon) \). In the following subsection, we propose two methods for designing $\mathbf{X}$, and subsequently, we derive the aforementioned performance metrics.
\subsection{Subspace Joint Beamforming (SJB)}\label{jointbeamforming}
In the SJB scenario, the signal transmitted by the BS, $\mathbf{X} \in \mathbb{C}^{N\times L}$, is defined as follows:
\begin{equation}
\mathbf{X}=\sqrt{p_t}\mathbf{w}\mathbf{s}\label{x},
\end{equation}
$\mathbf{w} \in \mathbb{C}^{N \times 1}$
is the beamforming vector, and $\mathbf{s} \in \mathbb{C}^{1 \times L}$ is the white Gaussian
signaling data stream
 for the user with unit power, $\frac{1}{L}\mathbf{s}\mathbf{s}^H\approx1$ when $L$ is sufficiently large \cite{CrameRaoBoundOptimizationforJoint}. $p_t$ is the transmit power of the BS. We assume that $\mathbf{w}$ lies in the complex span of the channel vector of the user and the target. Moreover, \textcolor{blue}{the beamforming vector is designed based on the estimated channel, $\mathbf{\bar{h}}$. Thus, at SJB, we have:
\begin{equation}
\mathbf{w} =\frac{b_1\mathbf{\bar{h}}+b_2\mathbf{a}}{\parallel b_1\mathbf{\bar{h}}+b_2\mathbf{a} \parallel}\label{wi},
\end{equation}}
where $b_1=|b_1|e^{j\phi_1}\in \mathbb{C}$ and $b_2=|b_2|e^{j\phi_2}\in \mathbb{C}$ are two design parameters which can be optimized to maximize the SINR of the user or to minimize the CRB. Fig. \ref{systemmodelfig} (A) illustrates the system model in this scenario.
\subsection{Linear Beamforming (LB)}
In LB, BS uses different beamforming for the user, \textcolor{blue}{$\mathbf{w}_1= \frac{\bar{\mathbf{h}}}{||\bar{\mathbf{h}}||}$} where $\mathbf{w}_1 \in \mathbb{C}^{N \times 1}$, and radar probing signal, $\mathbf{w}_2=\frac{\mathbf{a}}{||\mathbf{a}||}$ where $\mathbf{w}_2 \in \mathbb{C}^{N \times 1}$. The signal sent by the BS is:
\begin{equation}
\mathbf{X}=c_1\mathbf{w}_1\mathbf{s}_u+c_2\mathbf{w}_2\mathbf{s}_r,\label{xnew}
\end{equation}
 where $\mathbf{s}_u \in \mathbb{C}^{1 \times L}$ is the white Gaussian signaling data stream
 for the user, with unit power, $\frac{1}{L}E\{\mathbf{s}_u\mathbf{s}_u^H\}\approx1$ when $L$ is sufficiently large \cite{CrameRaoBoundOptimizationforJoint}. $\mathbf{s}_r \in \mathbb{C}^{1 \times L}$ is signal dedicated for radar with $\frac{1}{L}\mathbf{s}_r\mathbf{s}_r^H=1$.
The received signal at the user is:
\textcolor{blue}{
\begin{equation}
\mathbf{y}_u= c_1\frac{\mathbf{h}^H \mathbf{\bar{h}} \mathbf{s}_u}{\parallel \mathbf{\bar{h}} \parallel}+c_2
\frac{\mathbf{h}^H  \mathbf{a} \mathbf{s}_r}{\parallel \mathbf{a} \parallel}+ \mathbf{z}_u.\label{yunew}
\end{equation}}
where $\mathbf{z}_u\in \mathbb{C}^{1 \times L}$ is the AWGN vector where each of its elements has the distribution of the form $\mathcal{CN} (0,\sigma^2_u)$. Moreover, using  (\ref{xnew}), and orthogonality assumption between $\mathbf{s}_r$ and $\mathbf{s}_u$, the sample covariance matrix is:
\begin{align}
\mathbf{R}_x=\frac{1}{L}\mathbf {X}\mathbf{X}^H\approx |c_1|^2\mathbf{w}_1\mathbf{w}_1^H+ |c_2|^2\mathbf{w}_2\mathbf{w}_2^H.\label{rxnew}
\end{align}
We assume that the power of the BS is $p_t$, thus we have:
\begin{align}
p_t&=\text{Tr}(R_x)=|c_1|^2\text{Tr}(\mathbf{w}_1\mathbf{w}_1^H)+|c_2|^2\text{Tr}(\mathbf{w}_2\mathbf{w}_2^H)\nonumber\\
&\overset{(a)}=|c_1|^2\text{Tr}(\mathbf{w}_1^H\mathbf{w}_1)\!+|c_2|^2\text{Tr}(\mathbf{w}_2^H\mathbf{w}_2)\!\!\overset{(b)}=|c_1|^2+|c_2|^2.\label{power}
\end{align}
where (a) is due to $\text{Tr}(ab)=\text{Tr}(ba)$ and (b) is due to $\mathbf{w}_1^H\mathbf{w}_1=\frac{\mathbf{h}^H\mathbf{h}}{||\mathbf{h}||^2}=1$ and $\mathbf{w}_2^H\mathbf{w}_2=\frac{\mathbf{a}^H\mathbf{a}}{||\mathbf{a}||^2}=1$. Fig. \ref{systemmodelfig} (B) illustrates the system model in this scenario.
\subsection{Problem Formulation}
\textcolor{blue}{We denote the achievable C-R regions of the ISAC system as: $\mathcal{C} \triangleq \!\!\bigcup_{\boldsymbol{R}_{x} \succeq \boldsymbol{0}}\!\! \Big\{(\hat{R}, \hat{\Psi}) \ \Big|\!\! \ \hat{R} \leq \log_{2}\Big(1 + \frac{\boldsymbol{h}_{k}^{H} \boldsymbol{R}_{x} \boldsymbol{h}_{k}}{\sigma^{2}}\Big), \!\! \ \hat{\Psi} \geq \mathrm{CRB}(\boldsymbol{R}_{x}),
\ \mathrm{tr}(\boldsymbol{R}_{x}) \leq p_t \Big\}$ where \(\log_{2}\Big(1 + \frac{\boldsymbol{h}_{k}^{H} \boldsymbol{R}_{x} \boldsymbol{h}_{k}}{\sigma^{2}}\Big)\) represents the instantaneous rate at the user. Our goal is to determine the Pareto boundary of this region. Note that both the rate and CRB\((\boldsymbol{R}_{x})\) are convex with respect to \(\boldsymbol{R}_{x}\). Consequently, the C-R region is a convex set. Therefore, the entire Pareto boundary can be characterized by minimizing the CRB while varying the rate threshold, as shown in \cite{onthefundementaltradeoff} and \cite{MIMOIntegratedSensingandCommunicationCRBRateTradeoff}.
\\
For the case of SJB, we have \(\mathbf{R}_x = p_t \mathbf{w} \mathbf{w}^H\). Thus, the problem reduces to optimizing the beamforming vector \(\mathbf{w}\) to minimize the estimation CRB while ensuring a minimum rate \(\gamma\) under a maximum transmit power constraint. Thus, for each realization of the channels we have:
\begin{align}
&\underset{\mathbf{w}}{\min} \quad \mathrm{CRB}(\theta) \nonumber\\
& \text{s.t.} \quad \gamma \leq \log_{2}\Big(1 + \frac{|\mathbf{h}^H \mathbf{w}|^2}{\sigma^2_u}\Big), \nonumber\\
& \mathrm{tr}(p_t \mathbf{w} \mathbf{w}^H) \leq p_t, \label{optimization}
\end{align}
To solve this optimization problem, we use the following lemma, where it is proved in \cite[Lemma 1]{cramerraoboundoptimizationforjointradarcommunicationbeamforming}.}
\begin{lemma}
\textcolor{blue}{The optimal solution of \eqref{optimization} lies in the span of \(\{\mathbf{a}, \mathbf{h}\}\).}
\end{lemma}
\textcolor{blue}{It is important to note that by solving this optimization problem for a given rate threshold \(\gamma\), we can identify a point on the Pareto boundary of the C-R region. By varying the value of \(\gamma\), we can obtain the complete set of Pareto boundary points for different C-R tradeoffs.
\\
Considering this lemma, the beamforming vector for SJB is optimal among all beamforming strategies that use the same signal for both S \& C. However, although our proposed beamforming approach at LB lies in the span of \(\{\mathbf{a}, \mathbf{h}\}\), it may not be optimal among methods that use \textit{different signals} for S \& C. Nevertheless, it has the potential to outperform SJB. Additionally, LB offers the opportunity to fully exploit the DoF provided by MIMO radar \cite{JointTransmitBeamformingforMultiuser,optimaltransmitbeamformingintegrated}. Also, by employing DPC (discussed later), it can mitigate the radar signal at the user, potentially improving communication performance.}
\section{SJB performance analysis}\label{jointbeamformingperformance}
In this section, we derive S \& C performance metrics in SJB, as defined in Section \ref{systemmodel}.
\subsection{OP of the user}\label{opofuser}
\textcolor{blue}{
Based on (\ref{x}), replacing \(\mathbf{\bar{h}} = \mathbf{h} + \mathbf{e}\) at (\ref{wi}), $\mathbf{y}_u$ (derived in Section \ref{systemmodel}), and defining each element of the error vector $\mathbf{e}$ by $e_i\sim \mathcal{CN}(0,\sigma^2_e)$, the SINR at the user is:
\begin{align}
\text{SINR}\!\!&\overset{(a)}{=}\frac{p_t}{\sigma^2_u}\frac{(\sum_{i=1}^{N}x_i)^2+(\sum_{i=1}^{N}y_i)^2}{(\sum_{i=1}^{N}k_i)}\overset{(b)}{=}\frac{p_t}{\sigma^2_u}\frac{X^2+Y^2}{K},\label{sinr2i}
\end{align}
where ($a$) is due to defining RVs $x_i\triangleq \mathcal{R}(b_1|h_i|^2\!\!+b_1h^*_ie_i+b_2h^*_ie^{-jf_i})$, $y_i\triangleq\mathcal{I}(b_1|h_i|^2\!\!+b_1h^*_ie_i+b_2h^*_ie^{-jf_i})$, and $k_i\triangleq|b_1h_i+b_1e_i+b_2e^{-jf_i}|^2$, in which $\mathcal{R}$ and $\mathcal{I}$ indicates real and imaginary parts; ($b$) is obtained by defining $X=\sum_{i=1}^{N}x_i$, $Y=\sum_{i=1}^{N}y_i$ and $K=\sum_{i=1}^{N}k_i$. Therefore, the OP of the user, $P_u(\gamma))$, is: $P_u(\gamma)\overset{(a)}{=}\int_{0}^{\pi}P(\frac{p_t}{\sigma^2_u}\frac{X^2+Y^2}{K}<\gamma)|\theta)f_{\theta}(\theta)d\theta,$ where ($a$) follows from conditioning on \(\theta\) and the independence of \(\theta\), \(\mathbf{h}\), and \(\mathbf{e}\). Thus, to calculate the inner probability, we need to derive the joint PDF of $X$, $Y$, and $K$. We note that $X$, $Y$, and $K$ (also $x_i$, $y_i$, and $k_i$) are not independent, as they are functions of $h_i$, $e_i$, and $f_i$. By conditioning on $\theta$, $f_i$, $\forall i=1,...,N$, will be treated as constant in the following. We define $N$ random vectors, $\mathbf{d}_i=[x_i, y_i, k_i]^T \in \mathbb{R}^{3 \times 1}$, $\forall i=1,...,N$. For any pair of $j$ and $i\neq j$, the vectors $\mathbf{d}_j$ and $\mathbf{d}_i$ are i.i.d because $h_i$s and $e_i$s are i.i.d. Thus, by using multidimensional central limit theorem (CLT) \cite{enwiki:1155685628}, when $N$ is large \footnote{Section \ref{simulations} shows that for $N>9$, multidimensional CLT holds. Moreover, in \cite{Aunifiedperformanceframeworkfor} and with the help of simulation, the authors show that the accuracy of CLT for a one-dimensional RV holds for $N>8$.} (which holds for the case of MIMO ISAC, due to using large antenna arrays), we have $\sqrt{N}[\frac{1}{N}(\sum_{i=1}^{N}\mathbf{d}_i)-\mathbf{\mu_d}]\overset{d}{\rightarrow} \mathcal{N}_3(\mathbf{0},\mathbf{\Sigma_d})$, 
which means $\sum_{i=1}^{N}\mathbf{d}_i\overset{d} \rightarrow \mathcal{N}_3(N\mathbf{\mu_d}, N\mathbf{\Sigma_d})$, 
where $\mathbf{\mu_d}$ and $\mathbf{\Sigma_d}$ are mean vector and covariance matrix of $\mathbf{d}_i$ (the same for all $i=1,..., N$). Therefore, $[X, Y, K]^T\overset{(d)}{\rightarrow} \mathcal{N}_3(N\mathbf{\mu_d},N\mathbf{\Sigma_d})$\footnote{We observed that Mont-Carlo simulation confirmed the numerical approximation.}. Thus, by finding $\mathbf{\mu_d}$ and $\mathbf{\Sigma_d}$ with the aid of Lemma \ref{lemma1i} (proof at Appendix \ref{lemma1pi}), the joint PDF of $X$, $Y$, and $K$ is derived.
\begin{lemma}\label{lemma1i}
$\mathbf{\mu_d}$ and $\mathbf{\Sigma_d}$ are derived as $\mathbf{\mu_d}=[|b_1| \cos(\phi_1),|b_1| \sin(\phi_1),|b_1|^2 + |b_1|^2 \sigma^2_e + |b_2|^2]$ and (\ref{sigmadi}), respectively, where $\zeta \triangleq \sin(\phi_1)$, $\kappa \triangleq\cos(\phi_1)$, and $\delta=|b_1|^4+|b_1|^4\sigma^4_e+2|b_1|^2|b_2|^2+2|b_1|^4\sigma^2_e+2|b|^2_1b^2_2\sigma^2_e$.
\end{lemma}}
\begin{figure*}[t]
\normalsize
\textcolor{blue}{\begin{align}
\mathbf{\Sigma_d}\!\!=\!\!\begin{bmatrix}
|b_1|^2\kappa^2+\frac{1}{2}|b_1|^2\sigma^2_e+\frac{|b_2|^2}{2} & |b_1|^2\kappa\zeta & |b_1|^3\kappa+\kappa|b_1|^3\sigma^2_e+\kappa|b_1||b_2|^2\\
|b_1|^2\kappa\zeta & |b_1|^2\zeta^2+\frac{1}{2}|b_1|^2\sigma^2_e+\frac{|b_2|^2}{2} & |b_1|^3\zeta+\zeta|b_1|^3\sigma^2_e+\zeta|b_1||b_2|^2\\
|b_1|^3\kappa+\kappa|b_1|^3\sigma^2_e+\kappa|b_1||b_2|^2 & |b_1|^3\zeta+\zeta|b_1|^3\sigma^2_e+\zeta|b_1||b_2|^2& \delta
\label{sigmadi}
\end{bmatrix}.
\end{align}}
\hrulefill
\end{figure*}
We note that as derived in Lemma \ref{lemma1i}, the joint conditional PDF of $X$, $Y$, and $K$ is independent of $\theta$ and $\phi_2$. Using this and the assumption of uniform distribution for $\theta$, we have:
\begin{align}
\!\!\!\!\!P_u(\gamma)&=\!\!\!\iiint_{\!\!\!\frac{p_t}{\sigma^2_u}\!\!\frac{x^2+y^2}{k}<\gamma}\!\!\!\!\!\!\!\!\!\!\!\!\!\!\!\!\!\!\!f(x,y,k)\,dx\,dy\,dk,\label{outage2}
\end{align}
where $f(X, Y, K)$ is the PDF of a trivariate normal distribution with a mean vector of $N\mathbf{\mu_d}$ and a covariance matrix of $N\mathbf{\Sigma_d}$. In general, integrating a multivariate normal PDF over an arbitrary interval lacks a general analytical expression. Therefore, numerical methods such as ray tracing are necessary for its calculation \cite{Amethodtointegrate}. However, since the integral domain is quadratic, we simplify the quadratic form into a weighted sum of non-central chi-square RVs. Next, we follow an approach that helps us derive a closed-form expression for $P_u$. For simplicity, we assume $\sigma_e=0$.

At the user SINR defined in \ref{sinr2i}, we set $b_2=1$, while $b_1$ can take on any arbitrary complex number between $0$ and infinity. This allows the SINR to span a range similar to the scenario where both $b_1$ and $b_2$ were designed parameters. Subsequently, we express the SINR as a function of certain RVs, each independent of $|b_1|$, in contrast to $x_i$, $y_i$, $k_i$, $X$, $Y$, and $K$, which depended on both the phase and the amplitude of the designed parameters $b_1$ and $b_2$. By defining $r_i\triangleq \mathcal{R}(e^{jf_i}h_i)$, $t_i\triangleq \mathcal{I}(e^{jf_i}h_i)$, $\hat{k}_i\triangleq |h_i|^2$, $R\triangleq \sum_{i=1}^{N}r_i$, $T\triangleq \sum_{i=1}^{N}t_i$, $\hat{K}\triangleq \sum_{i=1}^{N}\hat{k}_i$ $\tilde{r}_i\triangleq \cos(f_i+\phi_1)m_i-\sin(f_i+\phi_1)n_i$, $\tilde{t}_i\triangleq \sin(f_i+\phi_1)m_i+\cos(f_i+\phi_1)n_i$, $\tilde{R}\triangleq \sum_{i=1}^{N}\tilde{r}_i$, $\tilde{T}\triangleq \sum_{i=1}^{N}\tilde{t}_i$, and after performing mathematical operations on (\ref{sinr2i}), the SINR of the user is: $\text{SINR}=\frac{p_t}{\sigma^2_u}\frac{\hat{K}^2|b_1|^2+R^2+T^2+2|b_1|\hat{K}\tilde{R}}{N+\hat{K}|b_1|^2+2|b_1|\tilde{R}}$. Using the CLT, the RVs \(\hat{K}\), \(R\), \(T\), and \(\tilde{R}\) are jointly Gaussian. By following the same approach as in Lemma \ref{lemma1i} (proof omitted due to space limitations), the mean and covariance matrix of these jointly Gaussian RVs are $\mathbf{\mu}_1=E\{[R, T, \hat{K}, \tilde{R}]^T\}=N[0, 0, 1, 0]^T$ and:
\begin{align}
\mathbf{\Sigma}_1&=\text{cov}\{[R, T, \hat{K}, \tilde{R}]^T\}=N\begin{bmatrix}
\frac{1}{2}& 0 & 0 & \frac{\kappa}{2}\\
0& \frac{1}{2} & 0 & \frac{-\zeta}{2}\\
0& 0 & 1 & 0\\
\frac{\kappa}{2}& \frac{-\zeta}{2} & 0 & \frac{1}{2}
\end{bmatrix}.\label{rtk}
\end{align}
Thus, \(P_u(\gamma)\) is obtained by integrating \(f(R, T, \hat{K}, \tilde{R})\) over the domain \(\text{SINR} < \gamma\). The closed form of this expression is derived in Lemma \ref{lemma6} (proof in Appendix \ref{lemma6p}).
\begin{lemma}\label{lemma6}
The OP of the user for SJB, \(P_u(\gamma)\) in (\ref{outage2}), when \(\sigma_e = 0\) is: $F_{w,k,\lambda,s,m}(0)$ where $m=-\tilde{D}_4(\frac{\tilde{a}_4}{2\tilde{D}_4})^2+|b_1|^2-N\frac{\gamma\sigma^2_u}{p_t}-\frac{\gamma\sigma^2_u}{p_t}|b_1|^2$, $s^2=\tilde{a}^2_1+\tilde{a}^2_3$, $w=[\tilde{D}_2 ,\tilde{D}_4]^T$, $k=[1,1]$, and $\lambda=(\frac{\tilde{a}_4}{2\tilde{D}_4})^2$. Here, $\tilde{D}_2$, $\tilde{D}_4$, $\tilde{a}_1$, $\tilde{a}_3$, and $\tilde{a}_4$ are $\frac{N}{2}$,$ \frac{N}{2}+N|b_1|^2$, $-\sqrt{N}|b_1|\frac{\gamma\sigma^2_u}{2p_t}$, $2|b_1|^3\sqrt{N}\frac{\gamma\sigma^2_u}{p_t}$, and $2|b_1|^3\sqrt{N}(2-\frac{\gamma\sigma^2_u}{p_t})+ 2|b_1|\sqrt{N}(1-\frac{\gamma\sigma^2_u}{p_t})$, respectively. We remark that $P_u(\gamma)$ is independent of $\phi_1$ (the phase of $b_1$)\footnote{One can calculate its parameters with the help of MATLAB toolbox provided by the authors of \cite{Amethodtointegrate}.}.
\end{lemma}
\subsection{OP of the target}\label{opoftarget}
We utilize the CRB, as derived in \cite[Appendix C]{TargetDetectionandLocalization}, for a given $\theta$, which is equal to (\ref{crb}) at the top of the next page, where $\mathbf{A}(\theta)= \mathbf{b}(\theta) \mathbf{a}^H(\theta)$; $ \mathbf{R}_x=\frac{1}{L}\mathbf {X}\mathbf{X}^H\approx p_t\mathbf{w}\mathbf{w}^H$ is the sample covariance matrix of $\mathbf{X}$. Inserting (\ref{x})- (\ref{wi}) into (\ref{crb}) and after some algebraic manipulation, the \text{CRB}($\theta$) is simplified as (\ref{cramer}), where $\mathbf{b'}$ denotes the derivative of $\mathbf{b}$ with respect to $\theta$;
\begin{align}
&\text{CRB}(\theta)=\frac{\sigma_r^2}{2 L p_t|\alpha|^2 || \mathbf{b'} ||^2 | \mathbf{a}^H \mathbf{w} | ^2}.\label{cramer}
\end{align}
\textcolor{blue}{By substituting (\ref{wi}) into (\ref{cramer}), we obtain (\ref{crameri}), where ($a$) follows from defining \(g(\theta) \triangleq \frac{6 \sigma_r^2}{L p_t|\alpha|^2 (M-1)(M)(M+1)\pi^2 \cos^2(\theta)}\), \(\tilde{X} \triangleq \sum_{i=1}^{N} \tilde{x}_i\), \(\tilde{Y} \triangleq \sum_{i=1}^{N} \tilde{y}_i\), and \(K \triangleq \sum_{i=1}^{N} k_i\), where $\hat{x}_i \triangleq \cos(f_i + \phi_1 + \tilde{\tilde{\phi_i}}) |b_1 \bar{h}_i|, \quad \hat{y}_i \triangleq \sin(f_i + \phi_1 + \tilde{\tilde{\phi_i}}) |b_1 \bar{h}_i|,$
and \(k_i\) is defined in subsection \ref{opofuser}. Here, \(\bar{h}_i \sim \mathcal{CN}(0, 1+\sigma^2_e)\) are the elements of \(\mathbf{\bar{h}}\), and \(\bar{h}_i = |\bar{h}_i|e^{j\tilde{\tilde{\phi_i}}}\), where \(|\bar{h}_i|\) follows a Rayleigh distribution with scale parameter \(\sqrt{1+\sigma^2_e}\), and \(\tilde{\tilde{\phi_i}}\) is uniformly distributed in \([- \pi, \pi)\).
\begin{figure*}[t]
\normalsize
\begin{align}
&\text{CRB}(\theta)=\frac{\sigma^2_R \text{Tr}(\mathbf{A}^H(\theta) \mathbf{A}(\theta) \mathbf{R}_x)}{2 \mid \alpha \mid ^2 L (\text{Tr}(\mathbf{A}^H(\theta) \mathbf{A}(\theta) \mathbf{R}_x) \text{Tr}(\mathbf{A}^{.H}(\theta) \mathbf{A}^.(\theta) \mathbf{R}_x)-\mid \text{Tr}(\mathbf{A}^{.H}(\theta) \mathbf{A}(\theta) \mathbf{R}_x)\mid^2)}. \label{crb}
\end{align}
\textcolor{blue}{\begin{align}
&\text{CRB}(\theta)\overset{(a)}{=}\frac{g(\theta)K}{N^2|b_2|^2+\tilde{X}^2+2N|b_2|\cos(\phi_2)\tilde{X}+\tilde{Y}^2+2N|b_2|\tilde{Y}\sin(\phi_2)},
\label{crameri}
\end{align}}
\textcolor{blue}{\begin{align}
&\mathbf{\tilde{\mu}}_d=[0 ,0 ,|b_2|^2+|b_1|^2(1+\sigma^2_e)]^T,\label{ssi}\\
&\mathbf{\tilde{\Sigma}}_d\!=\!\!\!\!\begin{bmatrix}
\frac{|b_1|^2(1+\sigma^2_e)}{2} & 0 & \!\!|b_1|^2(1+\sigma^2_e)|b_2|\cos(\phi_2)\\
0 & \frac{|b_1|^2(1+\sigma^2_e)}{2} & |b_1|^2(1+\sigma^2_e)|b_2|\sin(\phi_2)\\
|b_1|^2(1+\sigma^2_e)|b_2|\cos(\phi_2) & |b_1|^2(1+\sigma^2_e)|b_2|\sin(\phi_2)& \!\!|b_1|^4(1+\sigma^2_e)^2+2|b_1b_2|^2(1+\sigma^2_e)
\label{sigmadti}
\end{bmatrix}.
\end{align}}
\hrulefill
\end{figure*}
Therefore, to calculate $P_c(\epsilon)=P(\text{CRB}(\theta)>\epsilon)$, we need to derive the joint PDF of $\tilde{X}$, $\tilde{Y}$, and $K$. Following the same approach as subsection \ref{opofuser}, we define $N$ i.i.d. random vectors, $\mathbf{\tilde{d}}_i=[\tilde{x}_i, \tilde{y}_i, k_i]^T \in \mathbb{R}^{3 \times 1}$. When $N$ is large, we have:
$[\tilde{X},\tilde{Y}, K]^T\overset{(d)}{\rightarrow} \mathcal{N}_3(N\mathbf{\tilde{\mu}}_d,N\mathbf{\tilde{\Sigma}}_d)$,
where $\mathbf{\tilde{\mu}}_d$ and $\mathbf{\tilde{\Sigma}}_d$ are mean vector and covariance matrix of $\tilde{\mathbf{d}}_i$ (the same for $i=1,...,N$), respectively. Thus, by finding $\mathbf{\tilde{\mu}_d}$ and $\mathbf{\tilde{\Sigma}_d}$ with the aid of Lemma \ref{lemma2i} (the proof follows the same approach as used in the proof of Lemma \ref{lemma1i}), the joint PDF of $\tilde{X}$, $\tilde{Y}$, and $K$ is derived.
\begin{lemma}\label{lemma2i}
$\mathbf{\tilde{\mu}}_d$ and $\mathbf{\tilde{\Sigma}}_d$ are derived as (\ref{ssi}) and (\ref{sigmadti}), respectively.
\end{lemma}}
We note that as derived in Lemma \ref{lemma2i}, the joint conditional PDF of $\tilde{X}$, $\tilde{Y}$, and $K$ is independent of $\theta$ and $\phi_1$.

Assuming a uniform distribution for $\theta$ and defining the domain: $\mathcal{D}(\theta, \tilde{X},\tilde{Y}, K)= \frac{K}{\tilde{X}^2+\tilde{Y}^2+2N\mathcal{R}(b_2)\tilde{X}+2N\mathcal{I}(b_2)\tilde{Y}+N^2|b_2|^2}>\frac{\epsilon}{g(\theta)}$,
we have:
\begin{align}
P_{c}(\epsilon)&=\frac{1}{\pi}\!\!\int_{0}^{\pi}\!\!\iiint_{\mathcal{D}(\theta, \tilde{X},\tilde{Y}, K)}\!\!\!\!\!\!f(\tilde{X},\tilde{Y},K)\,d\tilde{X}\,d\tilde{Y}\,dK,d\theta,\label{outagesjb}
\end{align}
Here, \( f(\tilde{X}, \tilde{Y}, K) \) represents the PDF of a trivariate normal distribution with a mean vector of \( N\mathbf{\tilde{\mu}_d} \) and a covariance matrix of \( N\mathbf{\tilde{\Sigma}_d} \). Following the same approach as in Section \ref{opofuser}, we derive a closed-form expression for \( P_c \), as stated in Lemma \ref{lemma7}. Due to space limitations, the detailed derivation is omitted.
\begin{lemma}\label{lemma7}
OP of the target for SJB, $P_c(\epsilon)$ in (\ref{outagesjb}) when $\sigma_e=0$, is $\int_{0}^{\pi}F_{w(\theta),k,\lambda(\theta),s,m(\theta)}(0)d\theta$ where $w=|b_1|^2\frac{\epsilon}{2g(\theta)}$, $k=2$, $\lambda=2(\frac{N\frac{\epsilon}{g(\theta)}-1}{|b_1|\frac{\epsilon}{g(\theta)}}))^2$, $s=|b_1|^4$, $m=N-\frac{g(\theta)}{\epsilon}-|b_1|^2$.  We remark that $P_c(\epsilon)$ is independent of $\phi_1$ (the phase of $b_1$).
\end{lemma}
\begin{proof}
We write the domain of the integral in (\ref{outagesjb}) in a quadratic form as $|b_1|^2\frac{\epsilon}{2g(\theta)}((\sqrt{2}\tilde{R}+\sqrt{2}\frac{N\frac{\epsilon}{g(\theta)}-1}{|b_1|\frac{\epsilon}{g(\theta)}})^2+(\sqrt{2}\tilde{T})^2)-(K-1)|b_1|^2+N-\frac{g(\theta)}{\epsilon}-|b_1|^2<0$. The right-hand side of the inequality is in the form of a generalized chi-square RV with the mentioned parameters\footnote{Given that $|b_1|$ appears in the parameters of the generalized chi-square RVs in Lemma \ref{lemma6} and \ref{lemma7}, it is possible to determine $|b_1|$ at which $P_u=P_c$ through the utilization of numerical methods.}.
\end{proof}
\subsection{User Ergodic Rate and Target Ergodic CRB}
The ergodic rate of the user is:
\begin{align}
R=E[\log(1+\text{SINR})]\!\!&\overset{(a)}=\!\!\int_{0}^{\infty}\!\!\!\!P(\frac{p_t}{\sigma_t}\frac{X^2+Y^2}{K}>((2^t-1)))dt\nonumber\\
&\overset{(b)}=\int_{0}^{\infty} (1-P_u((2^t-1)))dt.\label{ergodicsjb}
\end{align}
where (a) and (b) are due to statistical properties and (\ref{outage2}), respectively. \textcolor{blue}{ The expectation is with respect to the joint distribution of \(X\), \(Y\), and \(K\)}. Therefore, $R$ can be derived from Lemma \ref{lemma6} by replacing $\gamma$ with $(2^t-1)$.

\textcolor{blue}{Next, we show that our proposed metric, which we refer to as the ergodic CRB, provides a tighter bound compared to Bayesian or deterministic CRBs.
\\
The minimum square error (MSE) in estimating \(\theta\) is expressed as follows:
\begin{align}
&E\{(\theta - \hat{\theta})^2\} = E_{\mathbf{h}}\{E_{\mathbf{Y}, \theta | \mathbf{h}} \{ (\theta - \hat{\theta})^2 \} \} \nonumber \\
&\overset{(a)}{\geq} E_{\mathbf{h}}\left\{ \frac{1}{-E_{\mathbf{Y}, \theta | \mathbf{h}} \left\{ \frac{\partial^2 \ln P(\mathbf{Y}, \theta | \mathbf{h})}{\partial \theta^2} \right\} } \right\}\nonumber\\
&\overset{(b)}{=} 
E_{\mathbf{h}}\left\{ \frac{1}{E_{\theta} \left\{ E_{\mathbf{Y} | \mathbf{h}, \theta} \left\{ - \frac{\partial^2 \ln P(\mathbf{Y}, \theta | \mathbf{h})}{\partial \theta^2} \right\} \right\}} \right\} \nonumber \\
&\overset{(c)}{=} E_{\mathbf{h}} \left\{ \frac{1}{E_{\theta} \{ J \} + E_{\theta} \left\{ - \frac{\partial^2 \ln P(\theta)}{\partial \theta^2} \right\} } \right\} \triangleq \text{BMSE} \label{bayesian}
\end{align}
Here, (a) follows from the BCRB \cite{detectionestimationandmodulation}, which offers a lower bound on the MSE for weakly unbiased estimators. (b) arises from the tower property of conditional expectations, and (c) stems from the independence of \(\theta\) and \(\mathbf{h}\). Additionally, \(J \triangleq E_{\mathbf{Y} | \mathbf{h}, \theta} \left\{ - \frac{\partial^2 \ln P(\mathbf{Y} | \theta, \mathbf{h})}{\partial \theta^2} \right\}\), as derived in \cite{TargetDetectionandLocalization}, represents the Fisher information matrix (FIM) when all parameters are fixed, i.e., \(E_{\mathbf{Y} | \mathbf{h}, \theta} \{ (\theta - \hat{\theta})^2 \} \geq \frac{1}{J}\).
\\
Moreover, ergodic CRB, is defined as \(\text{ECRB} = E_{\theta} \{ E_{\mathbf{h}} \{ \frac{1}{J} \} \}\). We first compute the inverse of the FIM, assuming all parameters are fixed (except for the noise in \(\mathbf{Y}\)). In the second step, we apply the expectation operators to obtain ECRB.
\\
Given the independence of \(\mathbf{h}\) and \(\theta\), we can express \(E_{\mathbf{h}, \theta} = E_{\theta} E_{\mathbf{h}} = E_{\mathbf{h}} E_{\theta}\), which leads to:
\begin{align}
&E\{(\theta - \hat{\theta})^2\} = E_{\theta} \{ E_{\mathbf{h}} \{ E_{\mathbf{Y} | \mathbf{h}, \theta} \{ (\theta - \hat{\theta})^2 \} \} \}\geq E_{\theta} \{ E_{\mathbf{h}} \{ \frac{1}{J}\} \}\nonumber\\
& = \text{ECRB} = E_{\mathbf{h}} \{ E_{\theta} \left\{ \frac{1}{J} \right\} \} \overset{a}{\geq} E_{\mathbf{h}} \left\{ \frac{1}{E_{\theta} \{ J \}} \right\}\nonumber \\
& \geq E_{\mathbf{h}} \left\{ \frac{1}{E_{\theta} \{ J \} + E_{\theta} \left\{ - \frac{\partial^2 \ln P(\theta)}{\partial \theta^2} \right\}} \right\} = \text{BMSE}
\end{align}
\\
Here, (a) holds because for any almost surely positive definite matrix \(B\), we have \(E\left( \frac{1}{B} \right) \geq \frac{1}{E\{B\}}\) [Lemma 3, \cite{someclassesofglobalcramerraobounds}]. Consequently, we demonstrate that the ECRB still serves as a lower bound for the estimation error and provides a tighter bound compared to Bayesian or deterministic CRBs.
\\
In summary, the CRB is expressed as a function of the random variables \(\mathbf{h}\) and \(\theta\). When the distributions of these variables are known, as in our case, we can evaluate \(E(\mathbf{CRB})\) and \(P(\mathbf{CRB} < \epsilon)\) to analyze the behavior of the CRB. This approach of treating the CRB as a random variable is utilized in \cite{cramerraoboundonaerospaceandelectronicsystems, RethinkingthePerformanceofISACSystem, aStatisticalCharacterizationofLocalizationPerformance, recentinsightsinthebayesiananddeterministic}.} Thus, ECRB is:
\begin{align}
\!\!\!\!\text{ECRB}=E[\text{CRB}]&\overset{(a)}=\int_{0}^{\infty}\!\!\!\!P(\text{CRB}>t)dt\overset{(b)}=\int_{0}^{\infty} \!\!\!\!P_c(t)dt,\label{ergodicsjbcrb}
\end{align}
where (a) and (b) are due to statistical properties and $P_c(\epsilon)=P(\text{CRB}(\theta)>\epsilon)$. \textcolor{blue}{The expectation is with respect to the joint distribution of \(\tilde{X}, \tilde{Y}\), and \(K\).} Therefore, $\text{ECRB}$ can be derived from Lemma \ref{lemma7} by replacing $\epsilon$ with $t$.
\section{LB Performance Analysis}\label{linearbeamforing}
In this section, we derive S \& C performance metrics in LB, as defined in Section \ref{systemmodel}.
\subsection{OP of the user}
\textcolor{blue}{In this subsection we derive the OP of the user without and with DPC at Lemma \ref{lemmanew2} and \ref{lemma3}, respectively. Due to the received signal given by (\ref{yunew}), the SINR of the user is:
\begin{align}
&\text{SINR}\!\!= \frac{|c_1|^2|\frac{\mathbf{h}^H \mathbf{\bar{h}} }{\parallel \mathbf{\bar{h}} \parallel}|^2}{\sigma_u^2+\frac{|c_2|^2}{||\mathbf{a}||^2}|\mathbf{a}^H\mathbf{h}|^2}\overset{((a))}{=}\frac{|c_1|^2(P^2+\tilde{Q}^2)}{(\sigma_u^2+\frac{|c_2|^2}{N}(R^2+T^2))U}\label{SINRnotdpc},
\end{align}}
\textcolor{blue}{(a) is due to $P\triangleq \sum_{i=1}^{N}p_i$,$\tilde{Q}\triangleq \sum_{i=1}^{N}\tilde{q}_i$,$U\triangleq \sum_{i=1}^{N}u_i$ where $r_i$, $t_i$, $R$ and $T$ are defined at subsection \ref{opofuser}, ${p}_i\triangleq \mathcal{R}(|h_i|^2+h^*_ie_i)$, $\tilde{q}_i\triangleq \mathcal{I}(|h_i|^2+h^*_ie_i)$, $u_i\triangleq |h_i+e_i|^2$. Therefore, the OP of the user is:
\begin{align}
&P_u(\gamma)=\int_{0}^{\pi}P(\frac{|c_1|^2(P^2+\tilde{Q}^2)}{(\sigma_u^2+\frac{|c_2|^2}{N}(R^2+T^2))U}<\gamma)|\theta)f_{\theta}(\theta)d\theta.
\label{outagenotdpc}
\end{align}
To calculate the inner probability, we need to derive the joint PDF of $R$, $T$, $P$, $\tilde{Q}$, $U$. Following the same approach as Section \ref{opofuser}, we have:
$[R, T, P, \tilde{Q}, U]^T\overset{(d)}{\rightarrow} \mathcal{N}_5(N\mathbf{\hat{{\mu}}}_d,N\mathbf{{\hat{\Sigma}}}_d)$, where $\mathbf{\hat{{\mu}}_d}$ and $\mathbf{\hat{{\Sigma}}_d}$ are derived at Lemma \ref{lemmanew} (the proof follows the same approach as used in the proof of Lemma \ref{lemma1i}).
\begin{lemma}\label{lemmanew}
 $\mathbf{\hat{{\mu}}}_d$ and $\mathbf{\hat{{\Sigma}}}_d$ are derived as the following.
\begin{align}
&\mathbf{{\hat{\mu}}}_d=[0 ,0 ,1,0,1+\sigma^2_e]^T,\\
&\mathbf{{\hat{\Sigma}}}_d\!=\!\!\!\!\begin{bmatrix}
\frac{1}{2} & 0 & 0  & 0 & 0\\
0 & \frac{1}{2} & 0 & 0 & 0\\
0 & 0& 1+\frac{\sigma^2_e}{2}&0 &1+\sigma^2_e\\
0 & 0 & 0  & \frac{\sigma^2_e}{2} & 0\\
0& 0 & 1+\sigma^2_e  & 0 & (1+\sigma_e^2)^2
\end{bmatrix}.
\end{align}
\end{lemma}
Thus, by having the joint PDF of $R$, $T$, $P$, $\tilde{Q}$, and $U$, $P_u(\gamma)$ is derived by integrating the PDF over $\mathcal{D}(\theta, R, T, P, \tilde{Q}, U)= \frac{|c_1|^2(P^2+\tilde{Q}^2)}{(\sigma_u^2+\frac{|c_2|^2}{N}(R^2+T^2))U}<\gamma$. Since the domain of the integral is quadratic, we calculate the closed form of $P_u(\gamma)$, derived in Lemma \ref{lemmanew2}. For simplicity we assumed $\sigma_e=0$.
\begin{lemma}\label{lemmanew2}
OP of the user for LB, $P_u(\gamma)$ in (\ref{outagenotdpc}), is $F_{w,k,\lambda,s,m}(0)$ where $w=-\frac{\gamma|c_2|^2}{2N}$, $k=2$, $\lambda=0$, $s=|c_1|^2$, $m=-\gamma\sigma_u^2$.
\end{lemma}
\begin{proof}
We write the domain of the integral  in a quadratic form as $K|c_1|^2-\gamma\sigma_u^2-\frac{\gamma|c_2|^2}{2N}((\sqrt{2}R)^2+(\sqrt{2}T)^2)<0$. The right-hand side of the inequality is in the form of a generalized chi-square RV with the mentioned parameters.
\end{proof}}
Moreover, when $\sigma_e=0$, Given that $\mathbf{s}_r$ is known to BS, we can apply the DPC theorem \cite{elgamal}\footnote{\textcolor{blue}{It can be shown that we cannot apply the DPC scheme from \cite{costa} in the case of imperfect CSI.}}. In fact, here, $\mathbf{s}_r$ acts as an additive interference (i.e., a state that is non-causal and known to the encoder).
\begin{lemma}\label{lemma3}
The instantaneous rate of the user when receiving the signal given by (\ref{yunew}) is equal to $\log_{2}(1+ \frac{|c_1|^2||\mathbf{h}||^2}{\sigma^2_u})$ when applying DPC.
\end{lemma}
\begin{proof}
Please refer to Appendix \ref{lemma3p}.
\end{proof}
Therefore the outage probability of the user by applying DPC is $P_{u}(\gamma)= P(\frac{|c_1|^2||\mathbf{h}||^2}{\sigma^2_u} <\gamma).$ Moreover, we have $|| \mathbf{h} ||^2 \sim \frac{1}{2}\chi(2N)$, where $\chi(N)$ is chi-squared distribution with $N$ degrees of freedom. Therefore, we have: $P_{u}(\gamma)= \xi(N,\frac{\gamma \sigma^2_u}{|c_1|^2})$, where $\xi$ is the regularized gamma function \cite{generalgamma}.
\subsection{OP of the Target}\label{opoftargetlb}
By replacing (\ref{rxnew}) in (\ref{crb}), the CRB($\theta$) is derived in the following lemma (proof in Appendix \ref{lemma4p}):
\textcolor{blue}{\begin{lemma}\label{lemma4}
CRB($\theta$) is given as:
\begin{align}
\text{CRB}(\theta)\!\!=\!\!\frac{\frac{Q}{||\mathbf{b}'||^2}\tilde{{K}} \psi}{\psi^2+\!\frac{|c_1c_2|^2}{||\mathbf{b}'||^2}MN\tilde{{K}}\big((\sum_{i=1}^{N}-f'_i\hat{t}_i)^2+(\sum_{i=1}^{N}f'_i\hat{r}_i)^2\big)},
\label{crbsimplified}
\end{align}
where $\hat{r}_i\triangleq \mathcal{R}(e^{jf_i}\bar{h}_i)=|\bar{h}_i|\cos(f_i+\tilde{\tilde{\phi}}_i)$, $\hat{t}_i\triangleq \mathcal{I}(e^{jf_i}\bar{h}_i)=|\bar{h}_i|\sin(f_i+\tilde{\tilde{\phi}}_i)$, $\tilde{{k}}_i\triangleq |\bar{h}_i|^2$, $\hat{R}\triangleq \sum_{i=1}^{N}\hat{r}_i$, $\hat{T}\triangleq \sum_{i=1}^{N}\hat{t}_i$, $\tilde{{K}}\triangleq \sum_{i=1}^{N}\tilde{{k}}_i$  and $Q=\frac{\sigma^2_R}{2 \mid \alpha \mid ^2 L}$, $||\mathbf{b}'||^2=\frac{\pi^2\cos^2(\theta)M(M^2-1)}{12}$ and $\psi \triangleq |c_1|^2\big(\hat{R}^2+\hat{T}^2\big)+|c_2|^2N \tilde{{K}}$.
\end{lemma}}
The problem here is more challenging compared to the one in SJB analysis. The reason is that following the method used in Section \ref{opoftarget}, reaches five RVs: $\hat{t}_i$, $\hat{r}_i$, $\tilde{{k}}_i$, $f'_i\hat{r}_i$, and $f'_i\hat{t}_i$. If we define $N$ random vectors as $[\hat{r}_i, \hat{t}_i, \tilde{k}_i, f'_i\hat{r}_i, f'_i\hat{t}_i]$ for $\forall i=1,..., N$ and calculate the covariance of these vectors, the elements of the covariance matrix of each vector depend on the index $i$. This implies that these random vectors are not i.i.d, so we cannot use multidimensional CLT. In this section, we propose three distinct methods to mitigate this challenge.\\
\subsubsection{Upper Bound}\label{upperbound}
If we ignore the term $\big((\sum_{i=1}^{N}-f'_i\hat{t}_i)^2+(\sum_{i=1}^{N}f'_i\hat{r}_i)^2\big)$ at the denominator of CRB($\theta$) in (\ref{crbsimplified}), we obtain an upper bound on CRB, named as UCRB. Thus,
\begin{align}
\!\!\!\!P_c(\epsilon)=P(\text{CRB}(\theta)>\epsilon)<P(\text{UCRB}(\theta)\!\!>\epsilon)\triangleq P_{Uc}(\epsilon).
\end{align}
Therefore, $P_{Uc}$ provides an upper bound on the OP of the target. Thus,
$P_{Uc}(\epsilon)=\int_{0}^{\pi}P(\frac{\frac{Q}{||\mathbf{b}'||^2}\tilde{K} }{\psi}>\epsilon)|\theta)f_{\theta}(\theta)d\theta$ and we need the joint PDF of $\hat{R}$, $\hat{T}$, and $\tilde{K}$. Following the same approach as in subsection \ref{opofuser}, we have: $[\hat{R},\hat{T}, \tilde{K}]^T\overset{(d)}{\rightarrow} \mathcal{N}_3(N\mathbf{\mu}_d,N\mathbf{\sigma}_d)$, where $\mathbf{\mu}_d$ and $\mathbf{\sigma}_d$ are derived at Lemma \ref{lemma5i} (the proof follows the same approach as the proof of Lemma~\ref{lemma1i}).
\textcolor{blue}{\begin{lemma}\label{lemma5i}
$\mathbf{\mu}_d$ and $\mathbf{\sigma}_d$ are derived as:
\begin{align}
\mathbf{\mu}_d=[0 ,0 ,1+\sigma^2_e]^T,
\mathbf{\sigma}_d\!=\!\!\!\!\begin{bmatrix}
\frac{1+\sigma^2_e}{2} & 0 & 0\\
0 & \frac{1+\sigma^2_e}{2} & 0\\
0 & 0& (1+\sigma^2_e)^2
\label{sigmadt2}
\end{bmatrix}.
\end{align}
\end{lemma}}
Assuming a uniform distribution for $\theta$ and defining $I(\mathcal{D})=\frac{1}{\pi}\!\!\int_{0}^{\pi}\!\!\iiint_{\mathcal{D}(\theta, \hat{R}, \hat{T}, \tilde{K})}\!\!f(\hat{R}, \hat{T}, \tilde{K})\,d\hat{R}\,d\hat{T}\,d\tilde{K}d\theta$, where $f(\hat{R}, \hat{T}, \tilde{K})$ represents the PDF of a trivariate normal distribution with a mean vector of $N\mathbf{\mu_d}$ and a covariance matrix of $N\mathbf{\sigma_d}$, we have $P_{Uc}(\epsilon)=I(\mathcal{D}_1)$, where $\mathcal{D}_1(\theta, \hat{R}, \hat{T}, \tilde{K})= \frac{Q\tilde{K}}{\psi}>||\mathbf{b}'||^2\epsilon$. One can utilize numerical methods such as ray-tracing \cite{Amethodtointegrate} to calculate this integral.
\subsubsection{Lower Bound}\label{lowerbound}
\textcolor{blue}{Using Cauchy–Schwarz inequality for $N$-dimensional complex space, we have:
\begin{align}
|\sum_{i=1}^{N}\!jf'_ie^{jf_i}\bar{h}_i|^2&<(\sum_{i=1}^{N}\!|jf'_i|^2)(\sum_{i=1}^{N}\!|e^{jf_i}\bar{h}_i|^2)=(\sum_{i=1}^{N}\!|f'_i|^2)(\tilde{{K}}),
\end{align}
where $\sum_{i=1}^{N}\!|f'_i|^2=\frac{\pi^2\cos^2(\theta)N(N^2-1)}{12}$. By replacing $|\sum_{i=1}^{N}\!jf'_ie^{jf_i}\bar{h}_i|^2$ with ${\tilde{K}}(\sum_{i=1}^{N}\!|f'_i|^2)$ at the denominator of CRB($\theta$) at (\ref{crbsimplified}), we obtain a lower bound on CRB, named as LCRB.} Thus,
\begin{align}
P_c(\epsilon)\!=\!\!P(\text{CRB}(\theta)\!>\epsilon)\!>\!P(\text{LCRB}(\theta)\!>\!\epsilon)\!\triangleq P_{Lc}(\epsilon),\label{ostad}
\end{align}
where $P_{Lc}$ is a lower bound of the OP of the target. Using (\ref{crbsimplified}) and (\ref{ostad}), we have: $P_{Lc}(\epsilon)\!\!=\!\!\!\int_{0}^{\pi}\!\!P(\frac{\frac{Q}{||\mathbf{b}'||^2}\tilde{K}\psi}{\psi^2+\!\frac{|c_1c_2|^2\tilde{K}^2N^2(N^2-1)}{M^2-1}}>\epsilon)|\theta)f_{\theta}(\theta)d\theta$.
Thus, by defining the domain: $\mathcal{D}_2(\theta, \hat{R},\hat{T}, \tilde{K})= \frac{Q\tilde{K}\psi}{\psi^2+\!\frac{|c_1c_2|^2K^2N^2(N^2-1)}{M^2-1}}>||\mathbf{b}'||^2\epsilon$, we have $P_{Lc}(\epsilon)=I(\mathcal{D}_2)$.
\subsubsection{Approximation}
\textcolor{blue}{When $N$ is large, due to the law of large numbers, we have:
\begin{align}
(\sum_{i=1}^{N}-f'_i\hat{t}_i)^2+(\sum_{i=1}^{N}f'_i\hat{r}_i)^2\overset{(p)}{\rightarrow} E\{(\sum_{i=1}^{N}-f'_i\hat{t}_i)^2+(\sum_{i=1}^{N}f'_i\hat{r}_i)^2 \}.
\end{align}
Thus, by replacing $(\sum_{i=1}^{N}-f'_i\hat{t}_i)^2+(\sum_{i=1}^{N}f'_i\hat{r}_i)^2\big)$ in the denominator of CRB($\theta$) in (\ref{crbsimplified}) with $E\{(\sum_{i=1}^{N}-f'_i\hat{r}_i)^2+(\sum_{i=1}^{N}f'_i\hat{t}_i)^2 \}\overset{(a)}=\sum_{i=1}^{N}(f'_i)^2E\{{\tilde{K}}\}=(1+\sigma^2_e)\sum_{i=1}^{N}(f'_i)^2$, where (a) is due to independence of $\hat{t}_i$ with $\hat{t}_j$ and $\hat{r}_i$ with $\hat{r}_j$, we will have an approximation of CRB($\theta$) named ACRB.}
Thus, we have:
\begin{align}
P_c(\epsilon)\!=\!P(\text{CRB}(\theta)\!>\!\epsilon)\!\approx\!\! P(\text{ACRB}(\theta)\!>\!\epsilon)\!\!\triangleq\! P_{Ac}(\epsilon),\label{ap}
\end{align}
where $P_{Ac}$ is an approximation of the OP of the target\footnote{In Section \ref{simulations}, we illustrate through numerical results that this approximation closely aligns with the exact value of $P_c$.}, which is: $P_{Ac}(\epsilon)\!=\!\!\!\int_{0}^{\pi}\!\!P(\frac{\frac{Q}{||\mathbf{b}'||^2}\tilde{K}\psi}{\psi^2+\!\frac{|c_1c_2|^2\tilde{K}N^2(N^2-1)}{M^2-1}}>\epsilon)|\theta)f_{\theta}(\theta)d\theta$.
Thus, by defining the domain: $\mathcal{D}_3(\theta, \hat{R},\hat{T}, \tilde{K})= \frac{Q\tilde{K}\psi}{\psi^2+\!\frac{|c_1c_2|^2\tilde{K}N^2(N^2-1)}{M^2-1}}>||\mathbf{b}'||^2\epsilon$, we have $P_{Ac}(\epsilon)=I(\mathcal{D}_3)$.

\textbf{Remark:} We note that if the signal sent by the BS is formulated as:
\begin{equation}
\mathbf{X} = c_1\frac{b_1\mathbf{h}+b_2\mathbf{a}}{\parallel b_1\mathbf{h}+b_2\mathbf{a} \parallel}\mathbf{s}_u + c_2\frac{b_1\mathbf{h}+b_2\mathbf{a}}{\parallel b_1\mathbf{h}+b_2\mathbf{a} \parallel}\mathbf{s}_r,
\end{equation}
we refer to this scenario as subspace linear beamforming (SLB). In this configuration, not only does the beamforming vector lie in the subspace created by both the channel of the user and the target, but the BS also utilizes different signals for the radar and user. The OP of the user in this scenario remains identical to that in SJB, except for the introduction of a coefficient $c_1$, which multiplies with $P_u$ in eq. (\ref{outage2}). Thus, the OP of the user increases. Furthermore, after applying the same methodology as detailed in Appendix \ref{lemma4p}, the OP of the target in this scenario is the same as the SJB scenario. Consequently, SJB outperforms SLB. When optimizing the OP of the user in SLB, the optimal $c_1$ becomes 1. In this particular case, SLB and SJB have equivalent performance.
\subsection{User Ergodic Rate and Target Ergodic CRB}
The ergodic rate of the user can be derived using (\ref{ergodicsjb}), Lemma \ref{lemma3} or Lemma \ref{lemmanew2}, and the ECRB can be derived using (\ref{ergodicsjbcrb}) and substituting $P_{Ac}$, derived at subsection \ref{opoftargetlb}. Due to space limitations, the details are omitted.

\section{Opportunistic Sensing/Communication Approaches}\label{specialpoints}
In this section, we derive the OP of the target and the user at boundary points in SJB and LB scenarios. These points represent scenarios where: 1) Sensing performance dominates, termed "opportunistic communication," achieved by setting $b_1=0$ and $|c_1|^2=0$ at SJB and LB, respectively. 2) Communication performance dominates, termed "opportunistic sensing," achieved by setting $b_1=\infty$ and $|c_1|^2=p_t$ at SJB and LB, respectively \footnote{In the LB case, $c_1$ and $c_2$ satisfy (\ref{power}). We set $|c_2|^2=p_t-|c_1|^2$, where $|c_1|^2$ can vary between $0$ and $\sqrt{p_t}$.}.
\subsection{Opportunistic Points of SJB}
First, we analyze opportunistic communication point. Based on (\ref{sinr2i}), the SINR of the user when $|b_1|$ goes to $0$ is $\frac{p_t}{\sigma^2_u}\frac{R^2+T^2}{N}$. Moreover, (\ref{rtk}) indicates that $R$ and $T$ are uncorrelated Gaussian random variables with variance $\frac{1}{2}$ and mean $0$. Therefore, $R^2+T^2$ follows a chi-square distribution with two degrees of freedom. Thus, the OP of the user is:
\begin{align}
P_u(\gamma)\!\!\!=\!P(\text{SINR}\!<\!\gamma)\!\!=\!\!P(R^2+T^2\!\!\!<\!\!\frac{N\gamma \sigma^2_u}{p_t})\!\!\!\overset{(a)}=\!\!\xi(1,\!\!\frac{N\gamma \sigma^2_u}{p_t}),
\end{align}
where (a) is based on the CDF of a Chi-square distribution, which is a regularized gamma function denoted by $\xi$ \cite{generalgamma}. Moreover, in this case, based on (\ref{crameri}), CRB is $\frac{g(\theta)}{N}$. Thus, since $\theta$ is uniformly distributed in the interval $[0,\pi]$, after calculating the CDF of $\cos^2(\theta)$ (a term at the denominator of $g(\theta)$), OP of the target is:
\begin{align}
P_c(\epsilon)\!\!&=\!\!P(\text{CRB}(\theta)>\epsilon)\!\!=\!\!\frac{\cos^{-1}(\sqrt{\frac{6 \sigma_r^2}{N\epsilon L p_t|\alpha|^2 (M-1)(M)(M+1)\pi^2 }})}{\pi},\label{crbb1zero}
\end{align}
when $\frac{6 \sigma_r^2}{N\epsilon L p_t|\alpha|^2 (M-1)(M)(M+1)\pi^2 }<1$, and $0$ otherwise.

\textcolor{blue}{Next, we analyze the opportunistic sensing point. When \(b_1\) approaches infinity, based on (\ref{sinr2i}), the SINR of the user will be \(\frac{p_t |||\mathbf{h}||^2 + |\mathbf{h}^H\mathbf{e}|^2}{(\sigma_u^2)(||\mathbf{h} + \mathbf{e}||^2)} = \frac{p_t(P^2 + \tilde{Q}^2)}{\sigma_u U}\). According to Lemma \ref{lemmanew}, \([P, \tilde{Q}, U]^T \overset{(d)}{\rightarrow} \mathcal{N}_3(N[\mathbf{\hat{\hat{\mu}}}_d]_{3:5}, N[\mathbf{\hat{{\Sigma}}}_d]_{3:5,3:5})\). Thus, by following the same approach as in Section IV-A, the OP of the user is derived by integrating this joint PDF over the domain \( \frac{p_t(P^2 + \tilde{Q}^2)}{\sigma_u U} < \gamma \). Under perfect CSI, this simplifies to:
\begin{align}
\!\!\!P_u(\gamma)\!\!\!=\!P(\text{SINR}\!<\!\gamma)\!\!=P(\parallel \mathbf{h} \parallel^2 <\frac{\gamma \sigma^2_u}{p_t})= \xi(N, \frac{\gamma \sigma^2_u}{p_t}).\label{userb1inf}
\end{align}
Furthermore, in this case, based on (\ref{crameri}) and the definition in Lemma \ref{lemma4}, the CRB is \(\frac{|\mathbf{\bar{h}}|^2 g(\theta)}{|\mathbf{a}^H\mathbf{\bar{h}}|^2} = \frac{g(\theta) {\tilde{K}}}{\hat{R}^2 + \hat{T}^2}\). Thus, the OP of the target is:
\begin{align}
P_c(\epsilon)&=P(\frac{(1+\sigma^2_e)\epsilon}{2g(\theta)}\left[\left(\sqrt{\frac{2}{1+\sigma^2_e}} \hat{R}\right)^2 + \left(\sqrt{\frac{2}{1+\sigma^2_e}} \hat{T}\right)^2\right]\nonumber\\
&- \left(\frac{\hat{\hat{K}}}{1+\sigma^2_e} - 1\right)(1+\sigma^2_e) < 0 )\overset{(a)}{=}F_{w,k,\lambda,s,m}(0).\label{crbb1inf}
\end{align}
where (a) follows from Lemma \ref{lemma5i}, showing that \(\hat{R}\), \(\hat{T}\), and \(\tilde{{K}}\) are independent Gaussian random variables. Hence, following \cite{general}, the parameters of this generalized chi-squared distribution are \(w = [\frac{(1+\sigma^2_e)\epsilon}{2g(\theta)}, \frac{(1+\sigma^2_e)\epsilon}{2g(\theta)}]\), \(k = [1, 1]\), \(\lambda = [0, 0]\), \(s = -(1+\sigma^2_e)\), and \(m = -(1+\sigma^2_e)\). }
\begin{figure*}
    \centering
    \subfigure[\textcolor{blue}{SJB}]{\includegraphics[scale=.4]{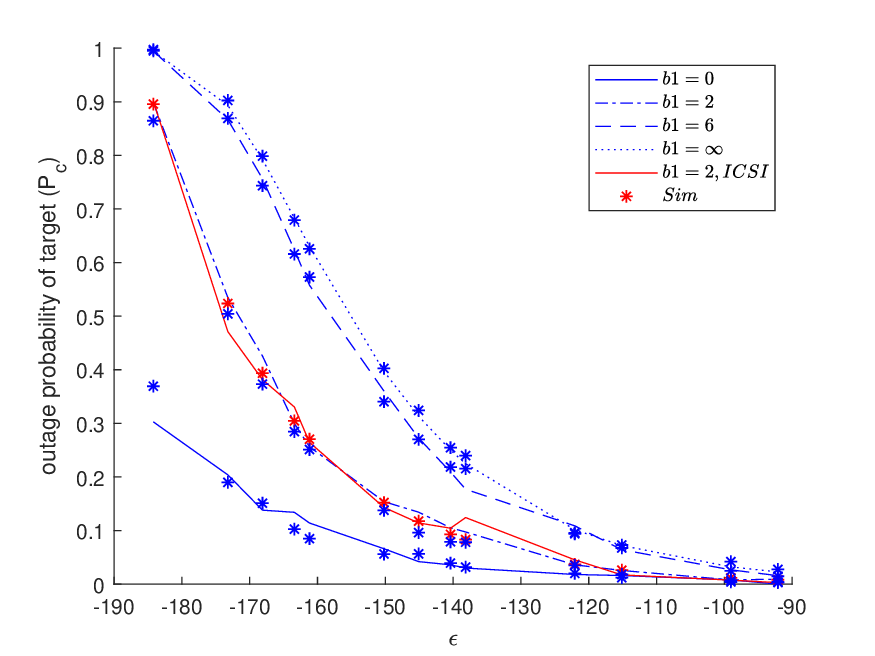}}\label{1}
    \subfigure[\textcolor{blue}{LB}]{\includegraphics[scale=.4]{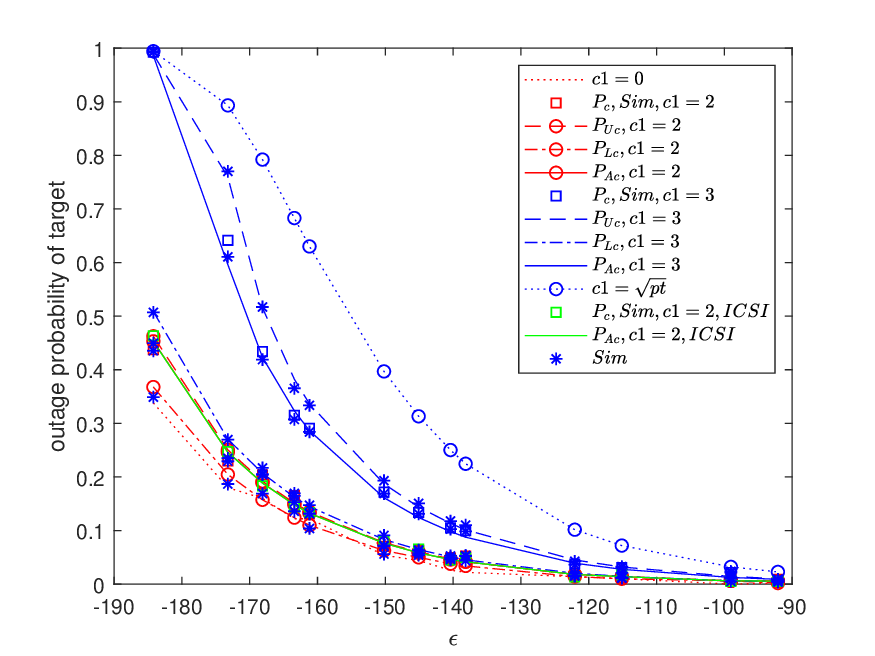}}\label{2}
    \subfigure[SJB and LB for $c_1=0.2, b_1=0.2$. ]{\includegraphics[scale=.4]{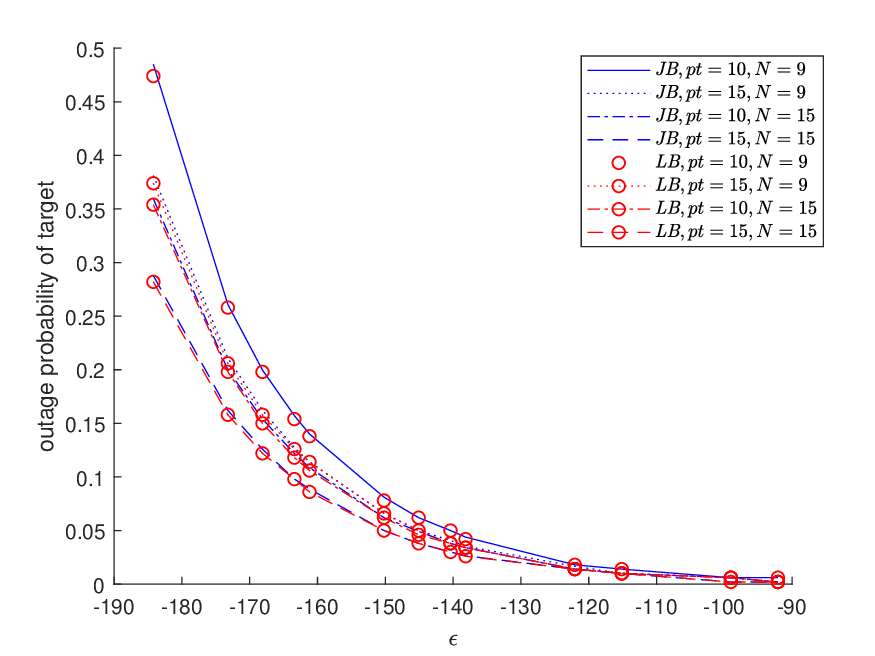}}\label{3}
    \caption{OP of the target versus $\epsilon(dB)$}
    \label{optarget}
\end{figure*}

\subsection{Opportunistic Points of LB}
First, we analyze opportunistic communication point. Based on Lemma \ref{lemma3} and (\ref{SINRnotdpc}), the SINR of the user when $|c_1|$ goes to $0$ is $0$, and the OP of the user is one. Moreover, in this case ($c_1=0$ and $|c_2|^2=p_t$), based on (\ref{crbsimplified}), CRB is $\frac{Q}{N p_t||\mathbf{b}'||^2}$. By substituting $Q$ and $||\mathbf{b}'||^2$ from Lemma \ref{lemma4} we have $\frac{Q}{N p_t||\mathbf{b}'||^2}=\frac{g(\theta)}{N}$, which is the same as CRB of the SJB when $b_1=0$. Thus, the OP of the target is the same as (\ref{crbb1zero}).

\textcolor{blue}{Next, we analyze the opportunistic sensing point. When \(|c_1|^2 = p_t\), based on (\ref{SINRnotdpc}), the SINR of the user will be \(\frac{p_t |||\mathbf{h}||^2 + |\mathbf{h}^H\mathbf{e}|^2}{(\sigma_u^2)(||\mathbf{h} + \mathbf{e}||^2)} = \frac{p_t(P^2 + \tilde{Q}^2)}{\sigma_u U}\), which is the same as the SINR for the user in the SJB case when \(b_1\) goes to infinity. Thus, the OP of the user in the LB is the same as the OP of the user in the SJB at the opportunistic sensing point. Moreover, based on (\ref{crbsimplified}), the CRB is \(\frac{Q}{p_t ||\mathbf{b}'||^2} \frac{{\tilde{K}}}{\hat{R}^2 + \hat{T}^2}\). Since \(\frac{Q}{p_t ||\mathbf{b}'||^2} = g(\theta)\), the OP of the target is the same as in the SJB scenario when \(b_1\) goes to infinity, as shown in (\ref{crbb1inf}).}

\textcolor{blue}{Remark: In perfect CSI, the opportunistic S \& C points for both SJB and LB (whether DPC is applied by the user or not) are identical. In the case of imperfect CSI, DPC cannot be applied. However, even in this scenario, the opportunistic S \& C points for both SJB and LB remain unchanged.}
\section{Numerical Results}\label{simulations}
Unless stated otherwise we consider the following:
A BS with $N=15$ and $M=17$ transmit and receive antennas;
the power budget $p_t=10$; the noise variance at the user and BS $\sigma_r=1$, and $\sigma_u=1$; the length of the radar frame and the reflection coefficient as $L = 30$ and $\alpha=1$, respectively; $\phi_1=\frac{\pi}{3}$ and $\phi_2=0$, representing the phases of $b_1$ and $b_2$, respectively. For the SJB scenario, we set $b_2=1$, and change $b_1$ between $0$ and infinity. In the LB case, we set $|c_2|^2=p_t-|c_1|^2$, and $|c_1|^2$ can be any number between $0$ and $\sqrt{p_t}$. The simulation results are based on $10000$ randomly seeded channel realizations. \textcolor{blue}{In the figure legends, whenever "$ICSI$" and "$NoDPC$" appear, they refer to imperfect CSI and the case where DPC is not applied, respectively.}

In Fig. \ref{optarget}.a and Fig. \ref{optarget}.b, the OP of the target in SJB for different $b_1$, and in LB for different $c_1$, are shown respectively, versus the threshold $\epsilon$ (OP thresholds for the target). \textcolor{blue}{As $\epsilon$ increases, the OP of the target decreases exponentially.} In SJB (LB), an increase in $b_1$ ($c_1$) results in an increase in the OP of the target, indicating a degradation in sensing performance. This is attributed to the fact that increasing $|b_1|$ ($c_1$) directs the beamforming vector toward the user, leading to higher error estimation for $\theta$. Moreover, in SJB, it is observed that for all $\epsilon$, after $b_1=10$, no significant change occurs compared to $b_1=\infty$, and the OP of the target is more sensitive to the change of $b_1$ in the interval $b_1=[0,10]$. Additionally, Fig. \ref{optarget}.b demonstrates that the OP of the target, $P_c$, lies between upper and lower bounds and is very close to the approximation approach, $P_{Ac}$. Also, by decreasing $c_1$, the upper bound, lower bound, and approximation approach become closer to each other, resulting in a tighter bound. Moreover, the simulation results in all figures confirm the numerical findings. \textcolor{blue}{Moreover, at SJB for \(b_1 = 2\) and at LB for \(c_1 = 2\), the OP of the target is plotted for both perfect and imperfect CSI. We can observe that the effect of channel errors on the OP of the target is small \footnote{\textcolor{blue}{To make the figures more readable, we have plotted imperfect CSI for only one set of parameters.}}.}

Fig. \ref{optarget}.c illustrates the OP of the target in SJB and LB versus $\epsilon$ for different $N$ and $p_t$ at $b_1=0.2, c_1=0.2$. It is worth noting that for LB, we have plotted $P_{AC}$ due to the observed matching of the real OP, $P_c$, to $P_{Ac}$ (from Fig. \ref{optarget}.b). The results indicate that an increase in $p_t$ leads to a decrease in the OP of the target due to the increment in the reflected power by the target. Furthermore, as the number of antennas $N$ increases, the OP of the target decreases, attributed to a more directed beamforming vector. Remarkably, we have observed that even with a reduction in the number of antennas to $N=9$, the simulation results closely match the analytical results, suggesting that the CLT still holds even for a small number of antennas (though not shown, the simulation results align with the analytical ones). Increasing $N$ and $p_t$ improves the system performance in the sensing aspect, with a slightly higher impact observed for an increase in $N$ compared to an increase in $p_t$. Additionally, the curves of LB approximately match those of SJB. This observation indicates that within a narrow range around the opportunistic communication points in SJB and LB, namely $b_1 \in [0,0.2]$ and $c_1 \in [0,0.2]$, the OP of the target behaves similarly in both scenarios. This extends the findings of Section \ref{specialpoints}, where it was demonstrated that at the opportunistic communication points in SJB and LB (corresponding to $b_1=0$ and $c_1=0$), the OP of the target exhibits similar behavior in both cases.

Fig. \ref{opuser}.a and Fig. \ref{opuser}.b depict the OP of the user at SJB for different $b_1$ and LB for different $c_1$, respectively, versus $\gamma$ (threshold of the OP of the user). With an increase in $\gamma$, the OP of the user increases, and at SJB (LB), an increase in $b_1$ ($c_1$) results in a decrease in the OP of the user, indicating an improvement in communication performance. This is because increasing $|b_1|$ ($c_1$) directs the beamforming vector towards the user. Additionally, at SJB, it is observed that the OP of the user is more sensitive to the change of $b_1$ in the interval $b_1=[0,4]$.  \textcolor{blue}{Moreover, at SJB for \(b_1 = 1\) and at LB for \(c_1 = 2\), the OP of the user is plotted for both perfect and imperfect CSI. It can be observed that imperfect CSI has a more detrimental effect on communication performance at SJB compared to LB. Additionally, by comparing Fig. \ref{opuser}.a with Fig. \ref{optarget}.a, we can see that channel errors have a greater impact on the OP of the user than on the OP of the target. Furthermore, in Fig. \ref{optarget}.b, by comparing the OP of the user when DPC is not applied at \(c_1 = 2\) to the case when it is applied, we observe that without DPC, communication performance degrades significantly, and even at low user OP thresholds, the OP of the user remains high.}

Fig. \ref{opuser}.c shows the OP of the user in SJB and LB versus $\gamma$ for different $N$ and $p_t$ at $b_1=100, c_1=3$ (boundary points, as observed in Fig. \ref{opuser}.a and Fig. \ref{opuser}.b, where no significant changes exist compared to $b_1=\infty, c_1=\sqrt{p_t}$). An increase in $p_t$ leads to a decrease in the OP of the user due to the increment in the received signal by the user. Moreover, as the number of antennas $N$ increases, the OP of the user decreases due to a more directed beamforming vector. An increase in $N$ by 50\% has a slightly higher impact on the OP of the user compared to an increase in $p_t$ by 50\%. By comparing Figure \ref{optarget}.c and Figure \ref{opuser}.c, we observe that the OP of the user is more sensitive to changes in $N$ or $p_t$ than the OP of the target. \textcolor{blue}{Additionally, when comparing the plot of LB with \(N = 15\) to that of SJB in Fig. 3(c), they approximately match. However, for the case where \(N = 9\), there is a slight deviation because \(c_1 = 3 < \sqrt{10}\). The closer \(c_1\) is to \(\sqrt{10}\), the closer the plots will be to each other.} This observation indicates that within a range around the opportunistic sensing points in SJB and LB, namely $b_1 \in [100,\infty]$ and $c_1 \in [3,\sqrt{p_t}]$, the OP of the user behaves similarly in both scenarios.

To illustrate the trade-off between S \& C, Fig. \ref{region}.a and Fig. \ref{region}.b depict the achievable region of the OP of the target and the OP of the user based on varying $|b_1|$ and $|c_1|$ when $N=15$, for different $\gamma$ and $\epsilon$. These figures reveal that by increasing $|b_1|$ ($|c_1|$) - i.e, moving to the left-hand side of the figure where $|b_1|$ ($|c_1|$) increases - the OP of the user decreases (improved communication performance), while the OP of the target increases (degraded sensing performance). Moreover, by increasing \(\gamma\), the region moves further from the origin, indicating a degradation in both S \& C performance. Additionally, Fig.\ref{region}.a demonstrates the benefits of utilizing ISAC compared to the time-sharing method. While ISAC enables the system to achieve skewed regions, \textcolor{blue}{ time-sharing methods result in straight lines (the points \((\alpha P_u(b_1=1, b_2=0) + (1 - \alpha), \alpha + (1 - \alpha)P_c(b_1=0, b_2=1))\), where \(\alpha \in [0,1]\) represents the fraction of time allocated to communication), leading to higher OPs for both the target and the user. Additionally, as depicted in Fig. \ref{region}.a, the LB performance degrades when DPC is not deployed, as the curve for LB without DPC is farther from the origin \footnote{\textcolor{blue}{To improve readability, in Fig. \ref{region}.b, time-sharing and LB without DPC are not shown. However, it was observed that whenever the region for LB with DPC falls below (or above) the region for SJB, LB without DPC exhibits the same behavior.}}.}

\textcolor{blue}{By comparing Fig.\ref{region}.a and Fig.\ref{region}.b, we observe that when the threshold of the OP of the target (\(\epsilon\)) is high (and OP of the target is low as observe in Fig.\ref{region}.b), the bottleneck is communication. In this case, JB performs better, meaning that the OP for both the target and the user is lower compared to LB, thus improving both S \& C performance. However, when \(\epsilon\) is low, depending on the region we are analyzing, the bottleneck is either communication or both S \& C. In this case, when communication is effective (i.e., the OP of the user is low), LB performs better. Specifically, we observe that as \(\epsilon\) increases, the intersection point of the LB and JB regions shifts to the left, and JB performs better (for both S \& C) than LB in a larger area of the region (with more pairs of points \((P_c, P_u)\)). At high \(\epsilon\), the JB curve completely falls below the LB curve, indicating that JB improves both S \& C performance across the entire region. However, as \(\epsilon\) decreases, LB performs better in a larger area of the region. Moreover, these figures demonstrate how a system designer can select the achievable S \& C OP thresholds and beamforming method for given OPs. For example, if we have $P_c=0.5$ and $P_u=0.2$, a set of \(\epsilon\) and \(\gamma\) values is achievable, which includes \(\epsilon = 10^{-7}, \gamma < 52\) for LB. However, with this pair of thresholds, SJB cannot meet the required S \& C OPs, and we must select another pair with \(\epsilon > 10^{-7}, \gamma < 52\). As another example, in Fig. \ref{region}.a, if we prioritize better sensing performance (i.e., if we need \(P_c < 0.4\)), we should use SJB. However, if we require a higher communication rate (i.e., \(P_u < 0.5\)), LB would be the better choice.}
\begin{figure*}
    \centering
    \subfigure[\textcolor{blue}{SJB}]{\includegraphics[scale=.4]{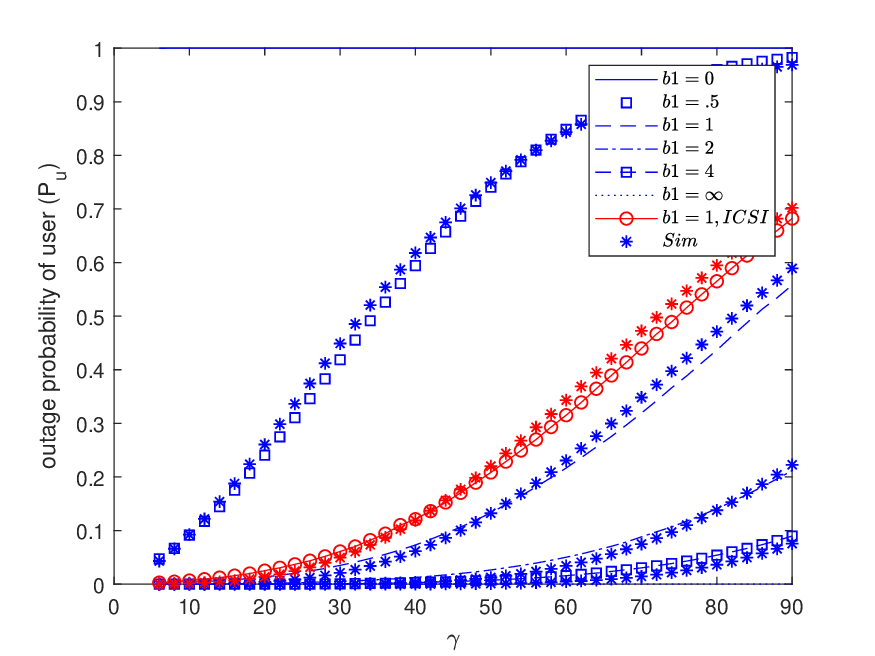}}\label{4}
    \subfigure[\textcolor{blue}{LB}]{\includegraphics[scale=.4]{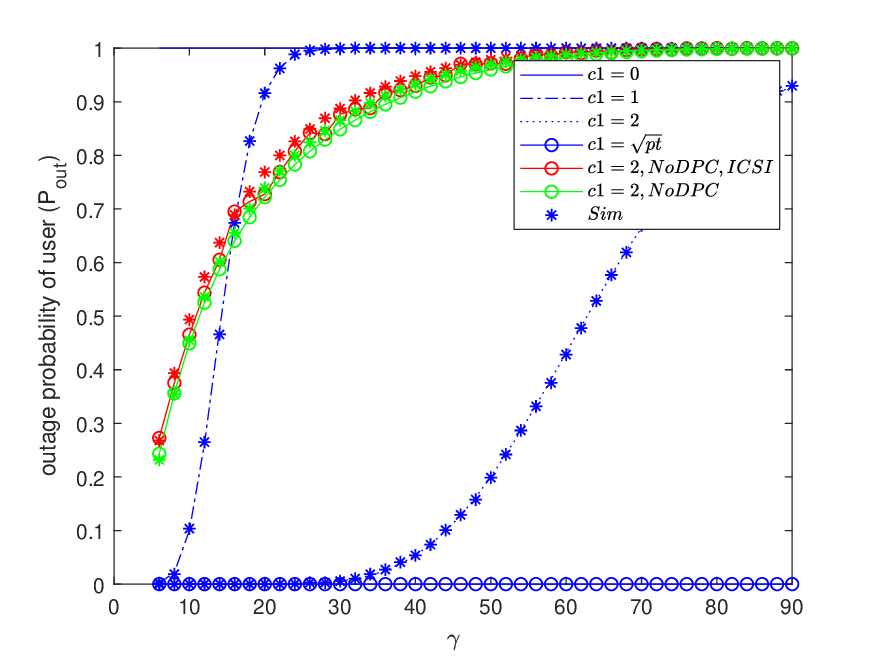}}\label{5}
    \subfigure[\textcolor{blue}{SJB and LB for $c_1=3, b_1=100$.}]{\includegraphics[scale=.4]{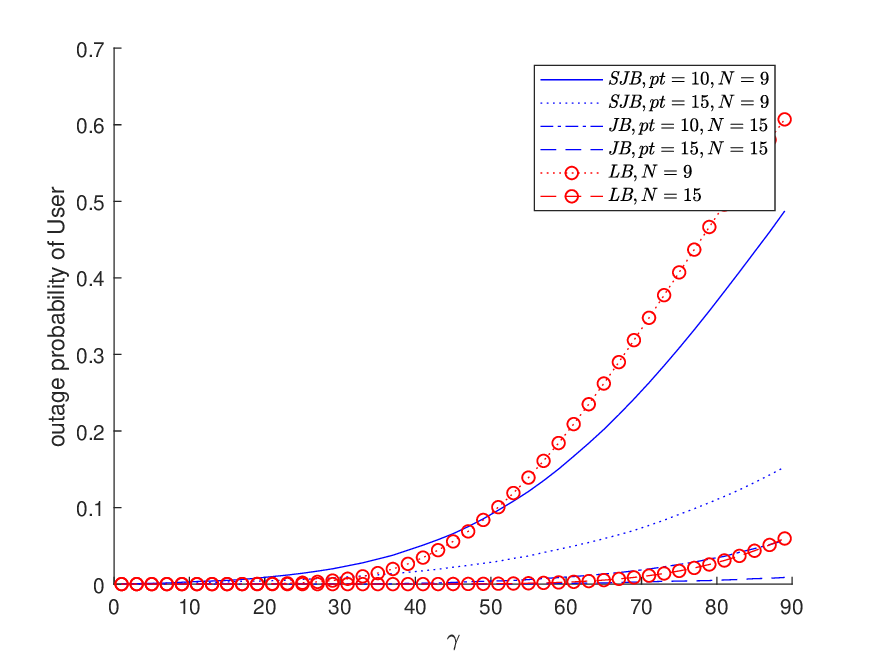}}\label{6}
    \caption{OP of the user versus $\gamma$}
    \label{opuser}
\end{figure*}
\section{Conclusion}\label{conclusion}
We investigated the fundamental performance limits of a downlink MIMO ISAC system considering the impact of randomness. The OP of the user and ergodic rate for communication and the OP of the target and ergodic CRB were derived as our performance metrics. We used two different beamforming methods to share resources between S \& C. Furthermore, by determining the Pareto optimal boundary of the performance metrics in both SJB and LB, we demonstrated that achieving optimal performance for S \& C simultaneously is highly unlikely. This observation suggests the existence of a fundamental S \& C trade-off in a random ISAC system.
\textcolor{blue}{We have demonstrated that under perfect CSI, the opportunistic S \& C points—where no power is allocated to one task (either sensing or communication), and not in the entire region—are the same for both SJB and LB, regardless of whether the user applies DPC. Under imperfect CSI, we showed that DPC cannot be applied. However, in this case, the opportunistic S \& C points for both SJB and LB remain the same. Moreover, when the OP threshold of the target, \(\epsilon\), increases, JB performs better (in both S \& C) than LB over a larger area of the C-R region (with more pairs of points \((P_c, P_u)\)). Conversely, as \(\epsilon\) decreases, LB performs better when communication is effective (i.e., the OP of the user is low).}

\bibliographystyle{ieeetr}
\bibliography{ref}
\begin{figure}
    \centering
    \subfigure[\textcolor{blue}{SJB and LB for low $\epsilon$}]{\includegraphics[scale=.4]{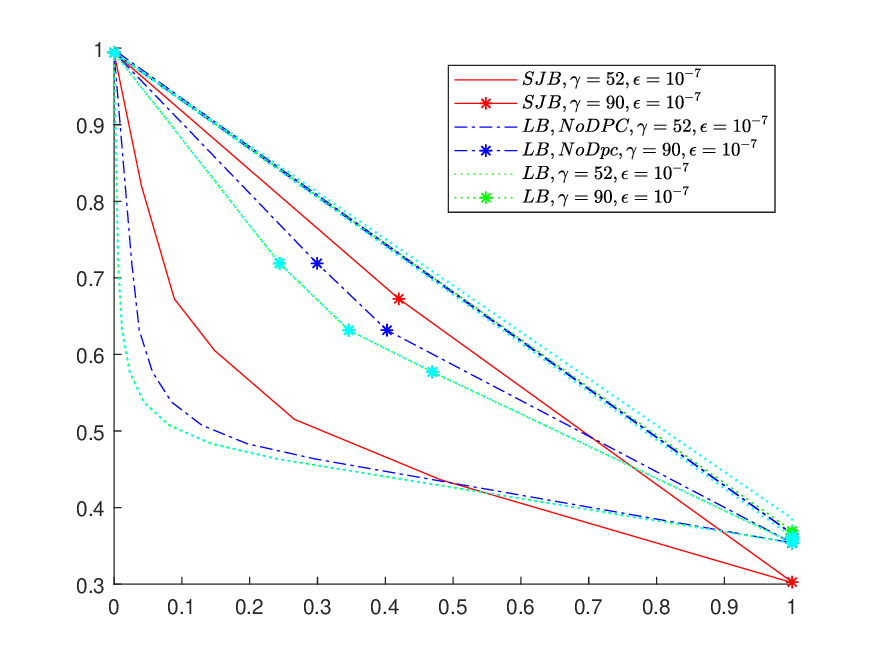}}\label{7}
    \subfigure[\textcolor{blue}{SJB and LB for high $\epsilon$}]{\includegraphics[scale=.4]{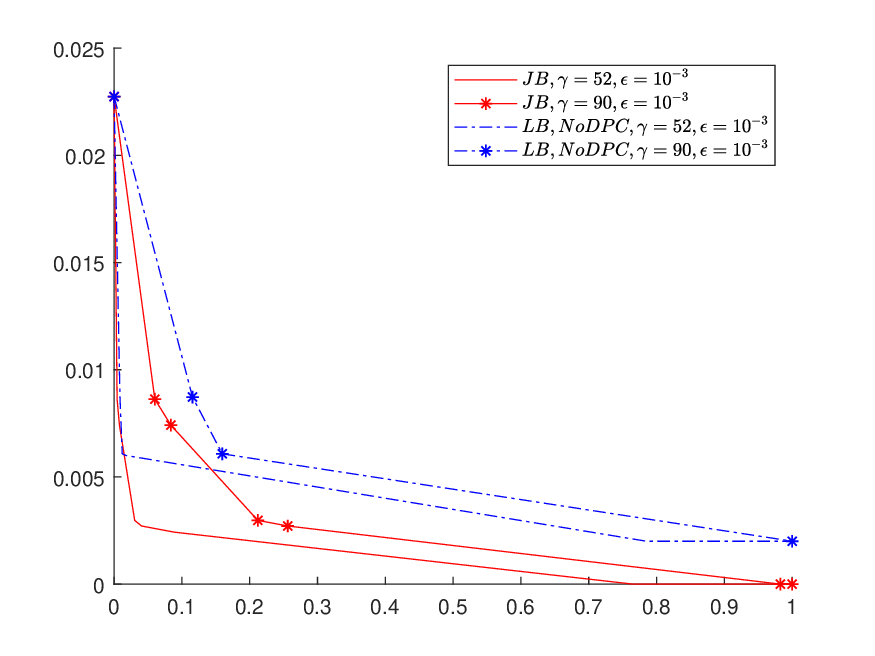}}\label{8}
    \caption{OP of the target versus OP of the user}
    \label{region}
\end{figure}
\appendices
\section{Proof of Lemma \ref{lemma1i}}\label{lemma1pi}
\textcolor{blue}{We first rewrite the random variables $x_i$, $y_i$, and $k_i$ as:
$x_i = |b_1| |h_i|^2 \cos(\phi_1) + |b_1 h_i e_i| \cos(\phi_1 - \phi_i + \tilde{\phi}_i) + |h_ib_2| \cos({\phi}_2-\phi_i - f_i)$, $y_i = |b_1| |h_i|^2 \sin(\phi_1) + |b_1 h_i e_i| \sin(\phi_1 - \phi_i + \tilde{\phi}_i) + |h_ib_2| \sin({\phi}_2-\phi_i - f_i)$, $k_i = |b_1 h_i|^2 + |b_1 e_i|^2 + |b_2|^2 + 2 |b_1^2 h_i e_i| \cos(\phi_i - \tilde{\phi}_i) + 2 |b_1b_2 h_i| \cos(f_i + \phi_1 + \phi_i-{\phi}_2) + 2 |b_1b_2 e_i| \cos(f_i + \phi_1 + \tilde{\phi}_i-{\phi}_2)$, where $b_1=|b_1|e^{j\phi_1}$, $b_2=|b_2|e^{j{\phi}_2}$, $h_i= |h_i|e^{j\phi_i} \sim \mathcal{CN}(0, 1)$, and $e_i= |e_i|e^{j\tilde{\phi}_i}\sim \mathcal{CN}(0, \sigma_e^2)$, with \( |e_i|\) and $|h_i|$ following a Rayleigh distribution with scale parameters \( \sigma_e \) and $1$, respectively. Moreover, $\tilde{\phi}_i$ and $\phi_i$ are uniformly distributed in $[-\pi, \pi)$. 
\\
First, we derive ${{\mathbf{\mu_d}}} = [E[{{x}}_i], E[{{y}}_i], E[{{k}}_i]]$. \( |h_i|^2 \) and \( |e_i|^2 \) follow exponential distributions with means 1 and $\sigma_e^2$, respectively. Thus, we have: $\mathbb{E}[|h_i|^2] = 1$ and $\mathbb{E}[|e_i|^2] = \sigma_e^2$. The last two terms in $x_i$ contain cosine functions involving the random variables \( \phi_i \) and \( \tilde{\phi}_i \). The expectations of these cosine terms will be zero because the angles are uniformly distributed. Therefore, $\mathbb{E}[x_i] = \mathbb{E}\left[ |b_1|| h_i|^2 \cos(\phi_1) \right] = |b_1| \cos(\phi_1)$. Using a similar approach, we obtain: $\mathbb{E}[y_i] = \mathbb{E}\left[ |b_1|| h_i|^2 \sin(\phi_1) \right] = |b_1| \sin(\phi_1)
$, $\mathbb{E}[k_i] = |b_1|^2 + |b_1|^2 \sigma_e^2 + |b_2|^2$. Thus, ${{\mathbf{\mu_d}}} = [|b_1| \cos(\phi_1), |b_1| \sin(\phi_1), |b_1|^2 + |b_1|^2 \sigma_e^2 + 1]$. Next, we calculate each element of the covariance matrix, $\mathbf{{{\Sigma_d}}}$. Due to space limitations, we omit the detailed derivations, but they are based on the following: Independence of $h_i$ and $e_i$, $E(|h_i|^4) = 2$, $E(|e_i|^4) = 2\sigma_e^4$, $E(\cos^2(\phi_1 - \phi_i + \tilde{\phi}_i)) = E(\sin^2(\phi_1 - \phi_i + \tilde{\phi}_i)) = \frac{1}{2}$, $E\left( \cos(\phi_1) \cdot \cos(\phi_1 - \phi_i + \tilde{\phi}_i) \right) = 0$
,\\ $E\left( \sin(\phi_1 - \phi_i + \tilde{\phi}_i) \cdot \cos(\phi_1 - \phi_i + \tilde{\phi}_i) \right) = 0$, $E(\cos(\phi_i + f_i)) = 0$. Moreover, using the identities $\cos(A)\cos(B) = \frac{1}{2} [\cos(A - B) + \cos(A + B)]$ and $\cos(A - B) = \cos(A)\cos(B) + \sin(A)\sin(B)$, we obtain: $E\left[\cos(\phi_i - \tilde{\phi}_i) \cos(\phi_1 - \phi_i + \tilde{\phi}_i)\right]
= \frac{1}{2} E\left[ \cos(-\phi_1 + 2\phi_i - 2\tilde{\phi}_i) + \cos(\phi_1) \right] = \frac{1}{2} \cos(\phi_1)$, $E(\cos(f_i + \phi_1 + \phi_i) \cos(\phi_i + f_i)) = \frac{1}{2} \cos(\phi_1)$, $E(\cos(f_i + \phi_1 + \phi_i) \sin(\phi_i + f_i)) = \frac{-1}{2} \sin(\phi_1)$, and $E\left[\cos(\phi_i - \tilde{\phi}_i) \sin(\phi_1 - \phi_i + \tilde{\phi}_i)\right] = \frac{1}{2} \sin(\phi_1).$}
\section{Proof of Lemma \ref{lemma6}}\label{lemma6p}
By following the same procedure as in \cite{Amethodtointegrate}, first, we write the domain of the integral, \(\text{SINR} < \gamma\), in a quadratic form, which is $\mathbf{u}^T\mathbf{Q}_2\mathbf{u}+\mathbf{q_1}^T\mathbf{u}+q_0<0$, where $\mathbf{u}=[R, T, K, \tilde{R}]^T$, $\mathbf{Q}_2=\begin{bmatrix}
1 & 0 & 0& 0\\
0 &1 & 0 & 0\\
0 & 0& |b_1|^2 &|b_1|\\
0 & 0& |b_1|& 0
\end{bmatrix}$, $\mathbf{q_1}=[0 , 0 ,-|b_1|^2\frac{\gamma\sigma^2_u}{p_t}, -2|b_1|\frac{\gamma\sigma^2_u}{p_t}]^T$, and $q_0=-N\frac{\gamma\sigma^2_u}{p_t}$. Then, we have $\mathbf{u}=\mathbf{S}\mathbf{r}+\mathbf{\mu}_1$, where $\mathbf{r}$ is a standard multivariate normal, and $\mathbf{S}=\mathbf{\Sigma}_1^{\frac{1}{2}}$. To calculate $S$ (symmetric square root of a positive semi-definite matrix $\mathbf{\Sigma}_1$), we decompose the matrix $\mathbf{\Sigma}_1$ into its eigenvalues and eigenvectors: $\mathbf{\Sigma}_1 = Q \Lambda Q^{-1}$ where $Q$ is an orthogonal matrix whose columns are the eigenvectors of $\mathbf{\Sigma}_1$, and $\Lambda$ is a diagonal matrix with the eigenvalues on the diagonal. Then, we have $S=\mathbf{\Sigma}_1^{1/2} = Q \Lambda^{1/2} Q^{-1}$. After doing some algebraic calculation the orthogonal elements of $\Lambda^{1/2}$ are $\sqrt{N},\sqrt{\frac{N}{2}}, 0,\sqrt{N}$ and we have $Q=\begin{bmatrix}
0 & \frac{\zeta}{\kappa} & -\kappa& \kappa\\
0 &1 & \zeta & -\zeta\\
1 & 0& 0 &0\\
0 & 0& 1& 1
\end{bmatrix}$, and $S=\frac{\sqrt{N}}{2}\begin{bmatrix}
\sqrt{2}\zeta^2+\kappa^2 & 2\zeta\kappa(\frac{1}{\sqrt{2}}-\frac{1}{2}) & 0& \kappa\\
2\zeta\kappa(\frac{1}{\sqrt{2}}-\frac{1}{2}) &\sqrt{2}\kappa^2+\zeta^2 & 0 & -\zeta\\
0 & 0& 2 &0\\
 \kappa &  -\zeta& 0& 1\label{s}
\end{bmatrix}$
Thus, $\mathbf{r}=\mathbf{S}^{-1}(\mathbf{u}-\mathbf{\mu}_1)$,
which results in a transformation to the integral domain as $\mathbf{r}^T\mathbf{\tilde{Q}}_2\mathbf{r}+\mathbf{\tilde{q_1}}^T\mathbf{r}+|b_1|^2-N\frac{\gamma\sigma^2_u}{p_t}-\frac{\gamma\sigma^2_u}{p_t}|b_1|^2<0$, where $\mathbf{\tilde{Q}}_2=\mathbf{S}\mathbf{Q}_2\mathbf{S}=\frac{N}{4}\begin{bmatrix}
\zeta^2+1 & \zeta\kappa & 2\kappa|b_1|& \kappa\\
\zeta\kappa &\kappa^2+1 & -2\zeta|b_1| & -\zeta\\
2\kappa|b_1| & -2\zeta|b_1|& 4|b_1|^2 &2|b_1|\\
\kappa & -\zeta& 2|b_1|& 1\label{Q2}
\end{bmatrix}$
and $\mathbf{\tilde{q_1}}=2\mathbf{S}\mathbf{Q}_2\mathbf{\mu}_1+\mathbf{S}\mathbf{q}_1=[|b_1|\sqrt{N}\kappa(1-\frac{\gamma\sigma^2_u}{p_t}), -|b_1|\sqrt{N}\zeta(1-\frac{\gamma\sigma^2_u}{p_t}),\sqrt{N}|b_1|^2(2-\frac{\gamma\sigma^2_u}{p_t}),\sqrt{N}|b_1|(1-\frac{\gamma\sigma^2_u}{p_t})]^T$. Next, by eigenvalue decomposition of $\mathbf{\tilde{Q}_2}=\mathbf{V}\mathbf{\tilde{D}}\mathbf{V}^T$, in which $\mathbf{V}$ is orthogonal, we use another transformation, $\mathbf{t}=\mathbf{V}^T\mathbf{r}$, which is also standard multivariate normal. After doing some algebraic calculation the orthogonal elements of $\mathbf{\tilde{D}}$ are $0, \frac{N}{2}, 0, \frac{N}{2}+N|b_1|^2$ and we have: $\mathbf{V}=\begin{bmatrix}
0 & \frac{\zeta}{\kappa} & -\kappa(4|b_1|^2+1)& \kappa\\
0 &1 & \zeta(4|b_1|^2+1) & -\zeta\\
\frac{-1}{2|b_1|} & 0& 2|b_1| &2|b_1|\\
1 & 0& 1& 1
\end{bmatrix}$. Therefore, it results in a transformation in the integral domain as $\mathbf{t}^T\mathbf{\tilde{D}}\mathbf{t}+\mathbf{\tilde{a}}^T\mathbf{t}+|b_1|^2-N\frac{\gamma\sigma^2_u}{p_t}-\frac{\gamma\sigma^2_u}{p_t}|b_1|^2=\sum_{i=2,4}^{}\tilde{D}_i\chi_{1,(\frac{\tilde{a}_i}{2\tilde{D}_i})^2}\!\!+\!\tilde{b}=\chi_{w,k,\lambda,s,m}<0$, where $\chi_{1,(\frac{\tilde{a}_i}{2\tilde{D}_i})^2}$ are chi-square RVs with 1 degree of freedom and non-centrality parameter $\frac{\tilde{a}_i}{2\tilde{D}_i}$. Also, 
$\mathbf{\tilde{a}}=\mathbf{V}^T\mathbf{\tilde{q_1}}=[-\sqrt{N}|b_1|\frac{\gamma\sigma^2_u}{2p_t}, 0, 2|b_1|^3\sqrt{N}\frac{\gamma\sigma^2_u}{p_t}, 2|b_1|^3\sqrt{N}(2-\frac{\gamma\sigma^2_u}{p_t})+ 2|b_1|\sqrt{N}(1-\frac{\gamma\sigma^2_u}{p_t})]^T$. $\tilde{D}_i$ are the diagonal element of $\mathbf{\tilde{D}}$ and $\tilde{a}_i$ are $i$-th elements of $\mathbf{\tilde{a}}$. Also $\tilde{b}\sim \mathcal{N}(m,s^2)$, where $m=-\tilde{D}_4(\frac{\tilde{a}_4}{2\tilde{D}_4})^2+|b_1|^2-N\frac{\gamma\sigma^2_u}{p_t}-\frac{\gamma\sigma^2_u}{p_t}|b_1|^2$ and $s^2=\tilde{a}^2_1+\tilde{a}^2_3$. Also, $\chi_{w,k,\lambda,s,m}$ is a generalized chi-square RV with parameters $w=[\tilde{D}_2 ,\tilde{D}_4]^T$, $k=[1,1]$, $\lambda=(\frac{\tilde{a}_2}{2\tilde{D}_2})^2+(\frac{\tilde{a}_4}{2\tilde{D}_4})^2$. Thus, $P_u(\gamma)$ is the CDF of $\chi_{w,k,\lambda,s,m}$ at $0$.
\section{Proof of Lemma \ref{lemma3}}\label{lemma3p}
By defining $\mathbf{y}_{u1}\triangleq c_1\frac{\mathbf{h}^H \textbf{h} \mathbf{s}_u}{\parallel \mathbf{h} \parallel}$ and $\mathbf{y}_{u2}\triangleq c_2
\frac{\mathbf{h}^H  \mathbf{a} \mathbf{s}_r}{\parallel \mathbf{a} \parallel}$, the received signal at the user given in (\ref{yunew}) is rewritten as:
\begin{equation}
\mathbf{y}_u= \mathbf{y}_{u1}+\mathbf{y}_{u2}+ \mathbf{z}_u.\label{ostad2}
\end{equation}
$\mathbf{y}_u$ in (\ref{ostad2}) is the same as an AWGN channel with additive interference, $\mathbf{y}_{u2}$. Since the interference ($\mathbf{y}_{u2}$) is known to the transmitter, we use the DPC \cite{elgamal} approach to cancel it at the receiver. The block diagram of the channel is depicted in Fig. \ref{channel}. 
\begin{figure}
\centering
    \includegraphics[scale=.45]{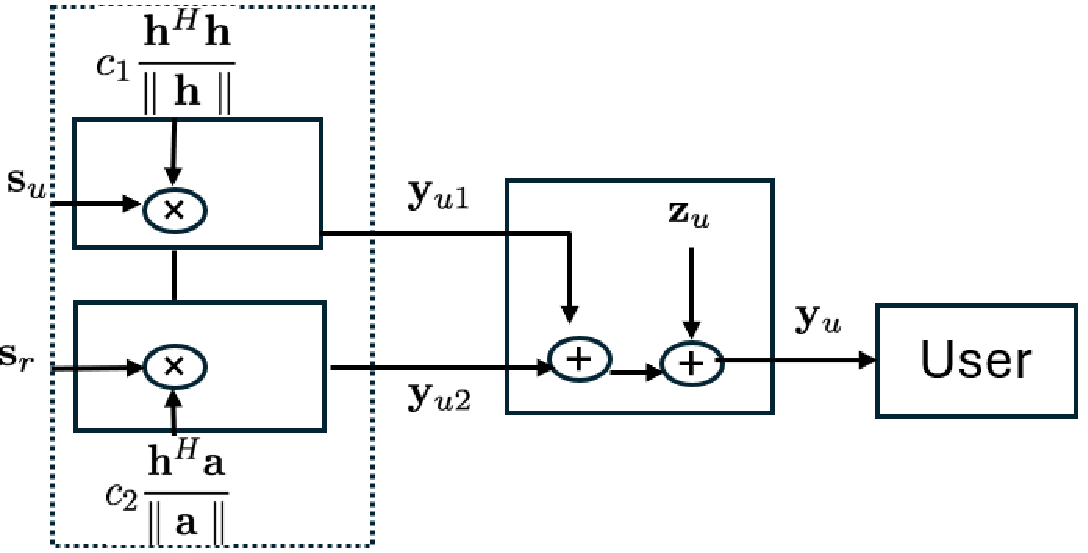}\caption{Communication channel model} \label{channel}
\end{figure}
We define an auxiliary random variable, $U=\mathbf{y}_{u1}+\alpha \mathbf{y}_{u2}$, where $\alpha$ will be derived later. Here, $\mathbf{y}_{u1}\sim N_c(0,|c_1|^2||\mathbf{h}||^2)$ (i.e., $\mathbf{s}_u \sim N_c(0,1)$), and it is independent of $\mathbf{y}_{u2}$ due to the independence assumption between $\mathbf{s}_u$ and $\mathbf{s}_r$; we can achieve:
\begin{align}
I(U,\mathbf{y}_u)-I(U,\mathbf{y}_{u2})=h(U|\mathbf{y}_{u2})-h(U|\mathbf{y}_{u}),\label{gilf}
\end{align}
where $I(.,.)$ and $h(.|.)$ denotes the mutual information and conditional differentiate entropy, respectively. We have:
\begin{align}
h(U|\mathbf{y}_{u2})\!=\!h(\mathbf{y}_{u1}\!+\!\alpha \mathbf{y}_{u2}|\mathbf{y}_{u2})\!\!\overset{(a)}=\!\!h(\mathbf{y}_{u1}|\mathbf{y}_{u2})\!\overset{(b)}= \!
h(\mathbf{y}_{u1})\label{hus}
\end{align}
where (a) is due to the property of $h(.)$ which states $h(x+f(y)|y)=h(x|y)$ when $f(y)$ is a function of $y$. (b) is due to $h(x|y)=h(x)$ if $x$ and $y$ are independent. Moreover, we have:
\begin{align}
h(U|\mathbf{y}_{u})&=h(\mathbf{y}_{u1}+\alpha \mathbf{y}_{u2}|\mathbf{y}_{u})=h(\mathbf{y}_{u1}+\alpha \mathbf{y}_{u2}-\alpha \mathbf{y}_{u} |\mathbf{y}_{u})\nonumber\\
&= h(\mathbf{y}_{u1}-\alpha(\mathbf{y}_{u1}+\mathbf{z}_{u})|\mathbf{y}_{u})\label{huy}
\end{align}
Thus, by replacing (\ref{hus}) and (\ref{huy}) into (\ref{gilf}), the achievable rate is:
\begin{align}
\!\!\!I(U\!,\mathbf{y}_u)\!\!-\!I(U\!,\mathbf{y}_{u2})\!\!=\!h(\mathbf{y}_{u1})\!\!-\!h(\mathbf{y}_{u1}-\!\!\alpha(\mathbf{y}_{u1}+\mathbf{z}_{u})|\mathbf{y}_{u})\label{rate}
\end{align}
If $\alpha$ is the weight of the linear minimum mean squared error (MMSE) estimate of $\mathbf{y}_{u1}$ when observing $\mathbf{y}_{u1}+\mathbf{z}_{u}$, then
$\mathbf{y}_{u1}-\alpha(\mathbf{y}_{u1}+\mathbf{z}_{u})$ is the MMSE error which is uncorrelated to the observation $\mathbf{y}_{u1}+\mathbf{z}_{u}$. Since they have also complex Gaussian distribution (the distribution of $\mathbf{y}_{u1}+\mathbf{z}_{u}$ is $N_c(0,\sigma^2_u+|c_1|^2||\mathbf{h}||^2)$), they are independent. Moreover, the MMSE estimation error is independent of $\mathbf{y}_{u2}$. Thus, $\mathbf{y}_{u1}-\alpha(\mathbf{y}_{u1}+\mathbf{z}_{u})$ is independent of $\mathbf{y}_{u}=\mathbf{y}_{u1}+\mathbf{y}_{u2}+ \mathbf{z}_u$. Therefore, we have:
\begin{align}
&h(\mathbf{y}_{u1}-\alpha(\mathbf{y}_{u1}+\mathbf{z}_{u})|\mathbf{y}_{u})=h(\mathbf{y}_{u1}-\alpha(\mathbf{y}_{u1}+\mathbf{z}_{u}))\nonumber\\
&\overset{(a)}=h(\mathbf{y}_{u1}\!\!-\!\alpha(\mathbf{y}_{u1}\!+\!\mathbf{z}_{u})|\mathbf{y}_{u1}+\mathbf{z}_{u})\!\!\overset{(b)}=h(\mathbf{y}_{u1}|\mathbf{y}_{u1}+\mathbf{z}_{u}\!),
\end{align}
where (a) is due to the independence of MMSE estimation error and the observation. (b) is due to the property of the entropy. Thus, the achievable rate at (\ref{rate}) becomes:
\begin{align}
&I(U,\mathbf{y}_u)-I(U,\mathbf{y}_{u2})=h(\mathbf{y}_{u1})-h(\mathbf{y}_{u1}|\mathbf{y}_{u1}+\mathbf{z}_{u})\nonumber\\
&\overset{(a)}=I(\mathbf{y}_u, \mathbf{y}_u+\mathbf{z}_{u})\overset{(b)}=h(\mathbf{y}_{u1}+\mathbf{z}_{u})-h(\mathbf{y}_{u1}+\mathbf{z}_{u}|\mathbf{y}_{u1})\nonumber\\
&\overset{(c)}=h(\mathbf{y}_{u1}+\mathbf{z}_{u})\!-h(\mathbf{z}_{u})\!\!\overset{(d)}=\!\!\log_{2}^{ \frac{\sigma^2_u+|c_1|^2||\mathbf{h}||^2}{\sigma^2_u}}\!\!\!\!\!\!\!\!\!\!\!\!\!\!=\log_{2}^{1+ \frac{|c_1|^2||\mathbf{h}||^2}{\sigma^2_u}},
\end{align}
where (a) and (b) are due to the definition of mutual information. (c) is due to the property of the entropy and the independence of $\mathbf{z}_{u}$ and $\mathbf{y}_{u1}$. (d) is due to the entropy of a complex Gaussian RV i.e. if $x\sim N_c(0,p)$ then $h(x)=\log_{2}^{e \pi p}$.

Next, we derive $\alpha$. By defining $\mathbf{\tilde{y}}=\mathbf{y}_{u1}+\mathbf{z}_u$ We have:
\begin{align}
\mathbf{\hat{y}}_{u1}=E\{\mathbf{y}_{u1}|\mathbf{\tilde{y}}\}=\int_{\mathbf{y}_{u1} \in \mathcal{C}}{} \mathbf{y}_{u1} f_{\mathbf{y}_{u1}|\mathbf{\tilde{y}}}(\mathbf{y}_{u1}|\mathbf{\tilde{y}})\label{conditionalmean}
\end{align}
Moreover, using the Bayesian theorem we have:
\begin{align}
&f_{\mathbf{y}_{u1}|\mathbf{\tilde{y}}}(\mathbf{y}_{u1}|\mathbf{\tilde{y}})=\frac{f_{\mathbf{\tilde{y}}|\mathbf{y}_{u1}}(\mathbf{\tilde{y}}|\mathbf{y}_{u1})f_{\mathbf{y}_{u1}}(\mathbf{y}_{u1})}{f_{\mathbf{\tilde{y}}(\mathbf{\tilde{y}})}}\nonumber\\
&\overset{(a)}=\frac{\frac{1}{\pi\sigma^2_u}e^{-\frac{|\mathbf{\tilde{y}}-\mathbf{y}_{u1}|^2}{\sigma^2_u}} 
\frac{1}{\pi|c_1|^2||\mathbf{h}||^2}e^{-\frac{|\mathbf{y}_{u1}|^2}{|c_1|^2||\mathbf{h}||^2}}}
{\frac{1}{\pi(\sigma^2_u+|c_1|^2||\mathbf{h}||^2)}e^{-\frac{|\mathbf{\tilde{y}}|^2}{\sigma^2_u+|c_1|^2||\mathbf{h}||^2}}}\nonumber\\
&=\frac{\sigma^2_u+|c_1|^2||\mathbf{h}||^2}{\pi\sigma^2_u|c_1|^2||\mathbf{h}||^2}e^{-\frac{\sigma^2_u+|c_1|^2||\mathbf{h}||^2}{\sigma^2_u|c_1|^2||\mathbf{h}||^2}|\mathbf{y}_{u1}-\frac{|c_1|^2||\mathbf{h}||^2}{\sigma^2_u+|c_1|^2||\mathbf{h}||^2}\mathbf{\tilde{y}}|^2,
}\label{conditionalpdf}
\end{align}
where (a) is due to the definition of the PDF of complex Gaussian distribution. From the conditional PDF derived at (\ref{conditionalpdf}), we have:
\begin{align}
\mathbf{y}_{u1}-\frac{|c_1|^2||\mathbf{h}||^2}{\sigma^2_u+|c_1|^2||\mathbf{h}||^2}\mathbf{\tilde{y}} \sim N_c(0,\frac{\sigma^2_u|c_1|^2||\mathbf{h}||^2}{\sigma^2_u+|c_1|^2||\mathbf{h}||^2})
\end{align}
when $\mathbf{\tilde{y}}$ is known. Thus, the conditional mean at (\ref{conditionalmean}) is:
\begin{align}
\mathbf{\hat{y}}_{u1}\!\!&=\frac{|c_1|^2||\mathbf{h}||^2}{\sigma^2_u+|c_1|^2||\mathbf{h}||^2}\mathbf{\tilde{y}}=\frac{|c_1|^2||\mathbf{h}||^2}{\sigma^2_u+|c_1|^2||\mathbf{h}||^2}(\mathbf{y}_{u1}+\mathbf{z}_u)
\end{align}
Thus, we have $\alpha=\frac{|c_1|^2||\mathbf{h}||^2}{\sigma^2_u+|c_1|^2||\mathbf{h}||^2}$.
\section{Proof of Lemma \ref{lemma4}}\label{lemma4p}
To simplify CRB, we begin by simplifying each term at the numerator and denominator of (\ref{crb}). We have:
\begin{align}
&\text{Tr}(\mathbf{A}^H(\theta) \mathbf{A}(\theta) \mathbf{R}_x)\overset{(a)}=\text{Tr}(\mathbf{a}\mathbf{b}^H\mathbf{b}\mathbf{a}^H (|c_1|^2\mathbf{w}_1\mathbf{w}_1^H+ |c_2|^2\mathbf{w}_2\mathbf{w}_2^H))\nonumber\\
&\overset{(b)}=|c_1|^2||\mathbf{b}||^2|\mathbf{a}^H\mathbf{w}_1|^2+|c_2|^2||\mathbf{b}||^2|\mathbf{a}^H\mathbf{w}_2|^2,\label{denom1}
\end{align}
where (a) is due to using $\mathbf{R}_X$ in (\ref{rxnew}), and (b) is due to $\text{Tr}(abc)=\text{Tr}(bca)=\text{Tr}(cab)$. Moreover, we have:
\begin{align}
&\text{Tr}(\mathbf{A}'^H(\theta) \mathbf{A}'(\theta) \mathbf{R}_x)=\nonumber\\
&\text{Tr}((\mathbf{a}\mathbf{b}'^H+\mathbf{a}'\mathbf{b}^H)(\mathbf{b}'\mathbf{a}^H+\mathbf{b}\mathbf{a}'^H) (|c_1|^2\mathbf{w}_1\mathbf{w}_1^H+ |c_2|^2\mathbf{w}_2\mathbf{w}_2^H))\nonumber\\
&\overset{(a)}=|c_1|^2(||\mathbf{b}'||^2|\mathbf{a}^H\mathbf{w}_1|^2+||\mathbf{b}||^2|\mathbf{a}'^H\mathbf{w}_1|^2)\nonumber\\
&+|c_2|^2(||\mathbf{b}'||^2|\mathbf{a}^H\mathbf{w}_2|^2+||\mathbf{b}||^2|\mathbf{a}'^H\mathbf{w}_2|^2),\label{denom2}
\end{align}
where (a) is due to the followig:
\begin{align}
&\!\!\!\!\!\text{Tr}((\mathbf{b}'\mathbf{a}^H+\mathbf{b}\mathbf{a}'^H)|c_1|^2\mathbf{w}_1\mathbf{w}_1^H(\mathbf{a}\mathbf{b}'^H+\mathbf{a}'\mathbf{b}^H))\nonumber\\
&\!\!\!\!\!\overset{(a)}=|c_1|^2(\text{Tr}(\mathbf{a}^H\mathbf{w}_1\mathbf{w}^H_1\mathbf{a}\mathbf{b}'^H\mathbf{b}')+\text{Tr}(\mathbf{a}'^H\mathbf{w}_1\mathbf{w}^H_1\mathbf{a}'\mathbf{b}^H\mathbf{b})).\\
&\!\!\!\!\!\text{Tr}((\mathbf{b}'\mathbf{a}^H+\mathbf{b}\mathbf{a}'^H)|c_2|^2\mathbf{w}_2\mathbf{w}_2^H(\mathbf{a}\mathbf{b}'^H+\mathbf{a}'\mathbf{b}^H))\nonumber\\
&\!\!\!\!\!\overset{(b)}=|c_2|^2(\text{Tr}(\mathbf{a}^H\mathbf{w}_2\mathbf{w}^H_2\mathbf{a}\mathbf{b}'^H\mathbf{b}')\!+\!\text{Tr}(\mathbf{a}'^H\mathbf{w}_2\mathbf{w}^H_2\mathbf{a}'\mathbf{b}^H\mathbf{b})).
\end{align}
where (a) and (b) are due to $\text{Tr}(abc)=\text{Tr}(bca)=\text{Tr}(cab)$ and $\text{Tr}(\mathbf{b}'\mathbf{a}^H\mathbf{w}_1\mathbf{w}_1^H\mathbf{a'}\mathbf{b}^H)=\text{Tr}(\mathbf{b}\mathbf{a}'^H\mathbf{w}_1\mathbf{w}_1^H\mathbf{a}\mathbf{b}'^H)=\text{Tr}(\mathbf{b}'\mathbf{a}^H\mathbf{w}_2\mathbf{w}_2^H\mathbf{a'}\mathbf{b}^H)=\text{Tr}(\mathbf{b}\mathbf{a}'^H\mathbf{w}_2\mathbf{w}_2^H\mathbf{a}\mathbf{b}'^H)=0$ (since we have:
$\mathbf{a}^H\mathbf{a}'=\mathbf{a}'^H\mathbf{a}=\mathbf{b}^H\mathbf{b}'=\mathbf{b}'^H\mathbf{b}=0$). Moreover:
\begin{align}
&\text{Tr}(\mathbf{A}'^H(\theta) \mathbf{A}(\theta) \mathbf{R}_x)=\nonumber\\
&\text{Tr}(\mathbf{b}\mathbf{a}^H(|c_1|^2\mathbf{w}_1\mathbf{w}_1^H+ |c_2|^2\mathbf{w}_2\mathbf{w}_2^H)(\mathbf{a}\mathbf{b}'^H+\mathbf{a}'\mathbf{b}^H))\nonumber\\
&\overset{(a)}=|c_1|^2||\mathbf{b}||^2\mathbf{a}^H\mathbf{w}_1\mathbf{w}_1^H\mathbf{a}'+|c_2|^2||\mathbf{b}||^2\mathbf{a}^H\mathbf{w}_2\mathbf{w}_2^H\mathbf{a}'\label{denom3}
\end{align}
where (a) is due to $\text{Tr}(\mathbf{b}\mathbf{a}^H\mathbf{w}_1\mathbf{w}_1^H(\mathbf{a}\mathbf{b}'^H))=\text{Tr}(\mathbf{b}\mathbf{a}^H\mathbf{w}_2\mathbf{w}_2^H(\mathbf{a}\mathbf{b}'^H))=0$, $\text{Tr}(abc)=\text{Tr}(bca)=\text{Tr}(cab)$ and because $\mathbf{a}^H\mathbf{w}_1\mathbf{w}_1^H\mathbf{a}'$ and $\mathbf{a}^H\mathbf{w}_2\mathbf{w}_2^H\mathbf{a}'$ are scalar. Therefore, by using (\ref{denom1}), (\ref{denom2}), and (\ref{denom3}), along with the following set of equations, the CRB($\theta$) is derived as (\ref{crbsimplified}):
\begin{align}
|c_1|^4||\mathbf{b}||^4\mathbf{a}^H\mathbf{w}_2\mathbf{w}_2^H\mathbf{a}'\mathbf{a}'^H\mathbf{w}_2\mathbf{w}_2^H\mathbf{a}=|c_2|^4|\mathbf{b}||^4|\mathbf{a}^H\mathbf{w}_2|^2|\mathbf{a}'^H\mathbf{w}_2|^2\nonumber\\
|c_1|^4||\mathbf{b}||^4\mathbf{a}^H\mathbf{w}_1\mathbf{w}_1^H\mathbf{a}'\mathbf{a}'^H\mathbf{w}_1\mathbf{w}_1^H\mathbf{a}=|c_1|^4||\mathbf{b}||^4|\mathbf{a}^H\mathbf{w}_1|^2|\mathbf{a}'^H\mathbf{w}_1|^2\nonumber
\end{align}
and $|c_1c_2|^2||\mathbf{b}||^4\mathbf{a}^H\mathbf{w}_2\mathbf{w}_2^H\mathbf{a}'\mathbf{a}'^H\mathbf{w}_1\mathbf{w}_1^H\mathbf{a}=0$ (since $\mathbf{w}_1=\frac{\mathbf{\bar{h}}}{||\mathbf{\bar{h}}||}$, $\mathbf{w}_2=\frac{\mathbf{a}}{||\mathbf{a}||}$, $\mathbf{a}^H\mathbf{a}'=0$), $|c_1c_2|^2||\mathbf{b}||^4\mathbf{a}^H\mathbf{w}_1\mathbf{w}_1^H\mathbf{a}'\mathbf{a}'^H\mathbf{w}_2\mathbf{w}_2^H\mathbf{a}=0$ (since $\mathbf{a}'^H\mathbf{a}=0$), $|c_1c_2|^2|\mathbf{b}||^4|\mathbf{a}^H\mathbf{w}_1|^2|\mathbf{a}'^H\mathbf{w}_2|^2=0$ (since $|\mathbf{a}'^H\mathbf{w}_2|^2=|\mathbf{a}'^H\frac{\mathbf{a}}{||\mathbf{a}||}|^2=0$, $Q\triangleq\frac{\sigma^2_R}{2 \mid \alpha \mid ^2 L}$, $\mathbf{w}_1=\frac{\mathbf{\bar{h}}}{||\mathbf{\bar{h}}||}$, $\mathbf{w}_2=\frac{\mathbf{a}}{||\mathbf{a}||}$, $||\mathbf{a}||^2=N$, $||\mathbf{b}||^2=M$, $|\mathbf{a}^H\mathbf{w}_2|^2=N$, and using the definitions at section (\ref{jointbeamformingperformance}).
\end{document}